\documentclass[12pt]{article}
\usepackage{geometry}
\usepackage{setspace}
\usepackage{fancyhdr}
\usepackage{amsfonts}
\usepackage{amsmath}
\usepackage{verbatim}
\usepackage{listings}
\usepackage{lscape}
\usepackage{graphics}
\usepackage{tabularx}
\usepackage{float}
\usepackage{breqn}
\usepackage{wasysym}
\usepackage{fix-cm}
\usepackage{bbm}
\usepackage{color}
\usepackage{amsthm}
\usepackage{enumitem}
\usepackage{authblk}

\usepackage{appendix}

\usepackage{graphicx}

\usepackage{multirow}
\usepackage{parskip}
\usepackage{amsmath, amsthm, amssymb}
\usepackage{dsfont}
\usepackage{bm}
\usepackage{parskip}

\usepackage{framed,enumitem}
\usepackage{tikz}
\usetikzlibrary{positioning}
\usepackage{accents}
\usepackage{comment}
\usepackage[ruled,vlined]{algorithm2e}
\usepackage{breqn}

\usepackage{natbib}

\newtheorem{mycor}{Corollary}
\newtheorem{myprop}{Proposition}
\theoremstyle{definition}
\newtheorem{mydef}{Definition}
\newenvironment{example}[1][Example]{\begin{trivlist}
\item[\hskip \labelsep {\bfseries #1}]}{\end{trivlist}}

\makeatletter
\newtheorem*{rep@theorem}{\rep@title}
\newcommand{\newreptheorem}[2]{%
\newenvironment{rep#1}[1]{%
 \def\rep@title{#2 \ref{##1}}%
 \begin{rep@theorem}}%
 {\end{rep@theorem}}}
\makeatother

\newreptheorem{theorem}{Theorem}
\newreptheorem{lemma}{Lemma}
\newreptheorem{mycor}{Corollary}
\newreptheorem{myprop}{Proposition}

\title{Partial Correlation Graphical LASSO}
\author[1,*]{Jack Storror Carter}
\author[2]{David Rossell}
\author[1,3]{Jim Q. Smith}
\affil[1]{Dept. of Statistics, University of Warwick, United Kingdom}
\affil[2]{Dept. of Business and Economics, Universitat Pompeu Fabra, Barcelona, Spain}
\affil[3]{The Alan Turing Institute, London, United Kingdom}
\affil[*]{Correspondence to j.s.carter@warwick.ac.uk}
\date{}

\begin{document}

\maketitle

\begin{abstract}
Standard likelihood penalties to learn Gaussian graphical models are based on regularising the off-diagonal entries of the precision matrix. Such methods, and their Bayesian counterparts, are not invariant to scalar multiplication of the variables, unless one standardises the observed data to unit sample variances. We show that such standardisation can have a strong effect on inference and 
introduce a new family of penalties based on partial correlations. We show that the latter, as well as the maximum likelihood, $L_0$ and logarithmic penalties are scale invariant. 
We illustrate the use of one such penalty, the partial correlation graphical LASSO, which sets an $L_{1}$ penalty on partial correlations. The associated optimization problem is no longer convex, but is conditionally convex. We show via simulated examples and in two real datasets that, besides being scale invariant, there can be important gains in terms of inference.
\end{abstract}

In Gaussian graphical models, most popular frequentist approaches to sparse estimation of the precision matrix penalise the absolute value of the entries of the precision matrix. 
Gaussian graphical models are invariant to scalar multiplication of the variables, however it is well-known that such penalisation approaches do not share this property.  
We show that the only scale-invariant strategies, within a large class of precision matrix penalties, are the logarithmic and $L_0$ penalties. 
It is possible to address this issue via a data preprocessing step of standardising the data to have unit sample variances.
However, as we illustrate next, this standardisation can adversely affect inference. In this paper we propose a family of methods based on partial correlations and show that they ensure scale invariance without requiring this standardisation step.  

As motivation we present a simple example where the goal is to estimate the entries in a $p \times p$ precision matrix $\Theta$. We set $p=50$ and generate $n=100$ independent Gaussian draws with zero mean and covariance $\Theta^{-1}$, where $\Theta$ follows the so-called star pattern, with $\theta_{ii}=1$ and $\theta_{i1}= \theta_{1i} = -1/\sqrt{p}$ for $i=1,\ldots,p$, and $\theta_{ij}=0$ otherwise.
This is a setting in which recovering the graphical model is relatively straightforward, see for example \cite{Yuan2007}.  
The top left panel in Figure \ref{fig:Star1} shows the regularisation path for the estimated partial correlations when applying GLASSO \citep{Friedman2008} to the unstandardised data. For a large range of values for the regularisation parameter $\rho$ the truly zero $\theta_{ij}$'s are completely separated from the non-zeroes.
However, the top right panel shows that when standardising the data to unit sample variances the quality of the inference suffers. In particular the true graphical model is not recovered for any $\rho$. The bottom left panel shows the results obtained by applying a LASSO penalty to the partial correlations, our proposed PC-GLASSO, which as we show is scale invariant. The bottom right panel demonstrates how the estimation accuracy measured by Kullback-Leibler loss (see Section \ref{sec:Simulations}) of GLASSO and two other methods reviewed below suffer in comparison to PC-GLASSO when using standardised data.

\begin{figure}[ht]
\vspace{10pt}
	\centering
	\includegraphics[scale=0.55]{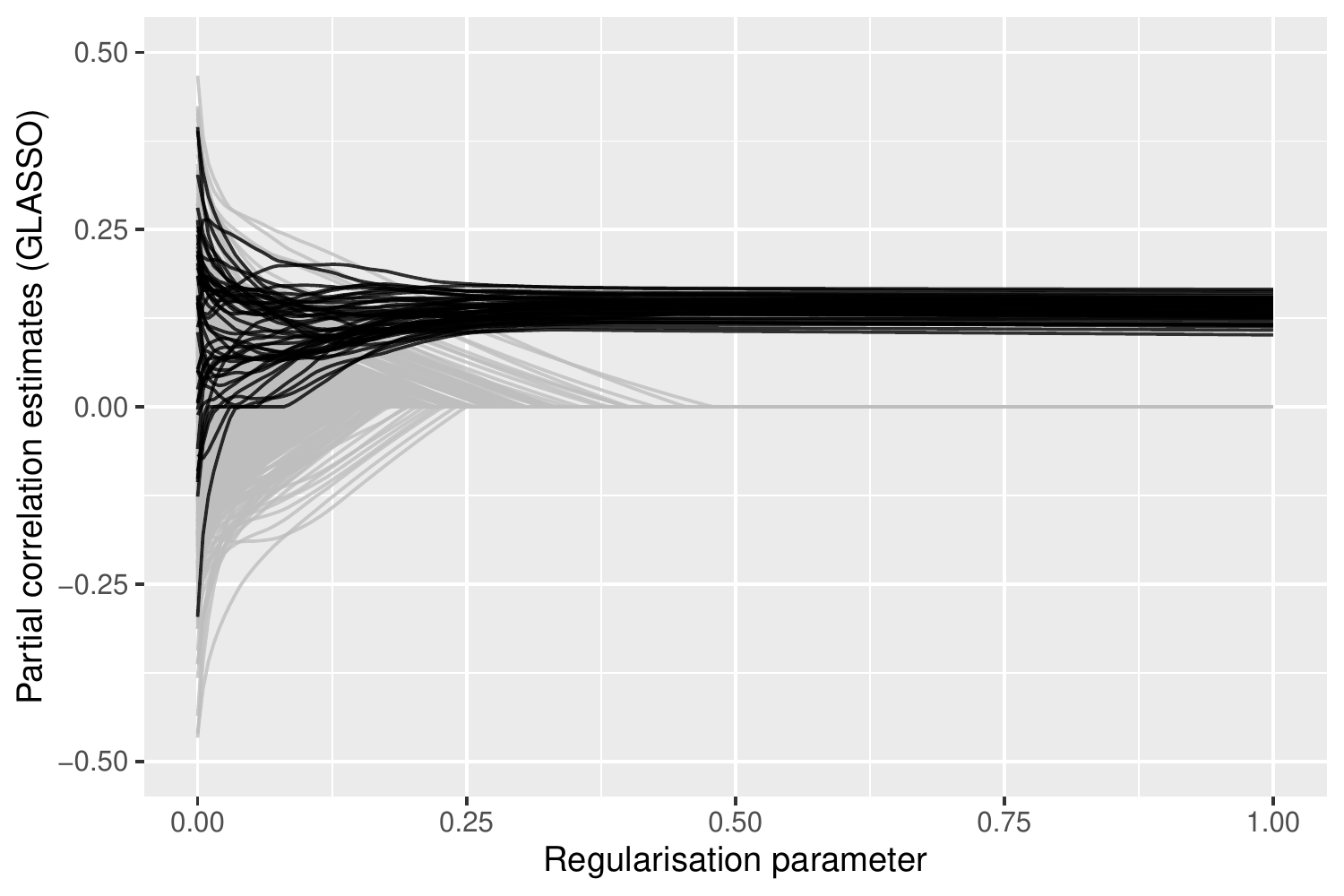}
	\centering
	\includegraphics[scale=0.55]{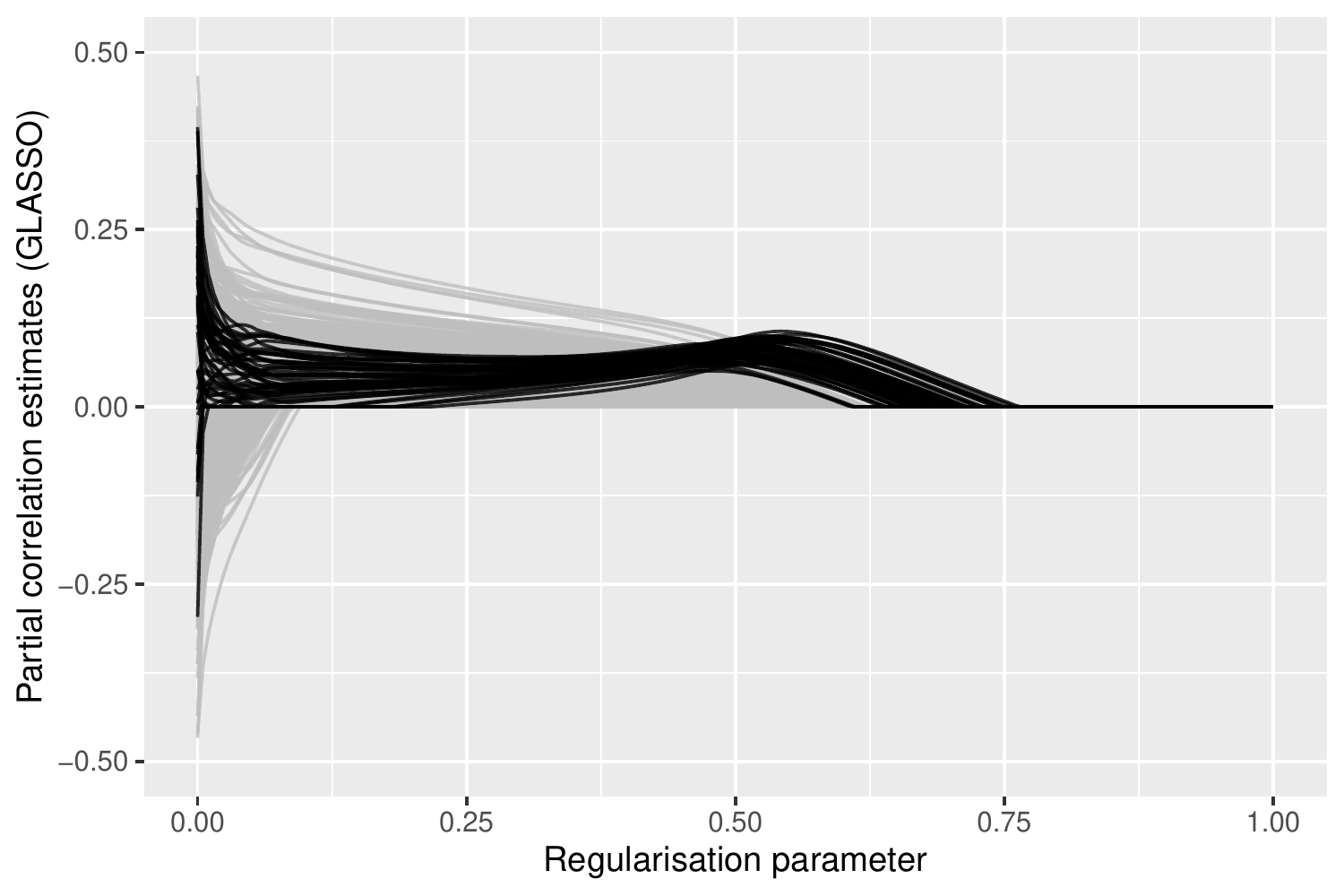}
	\centering
	\includegraphics[scale=0.55]{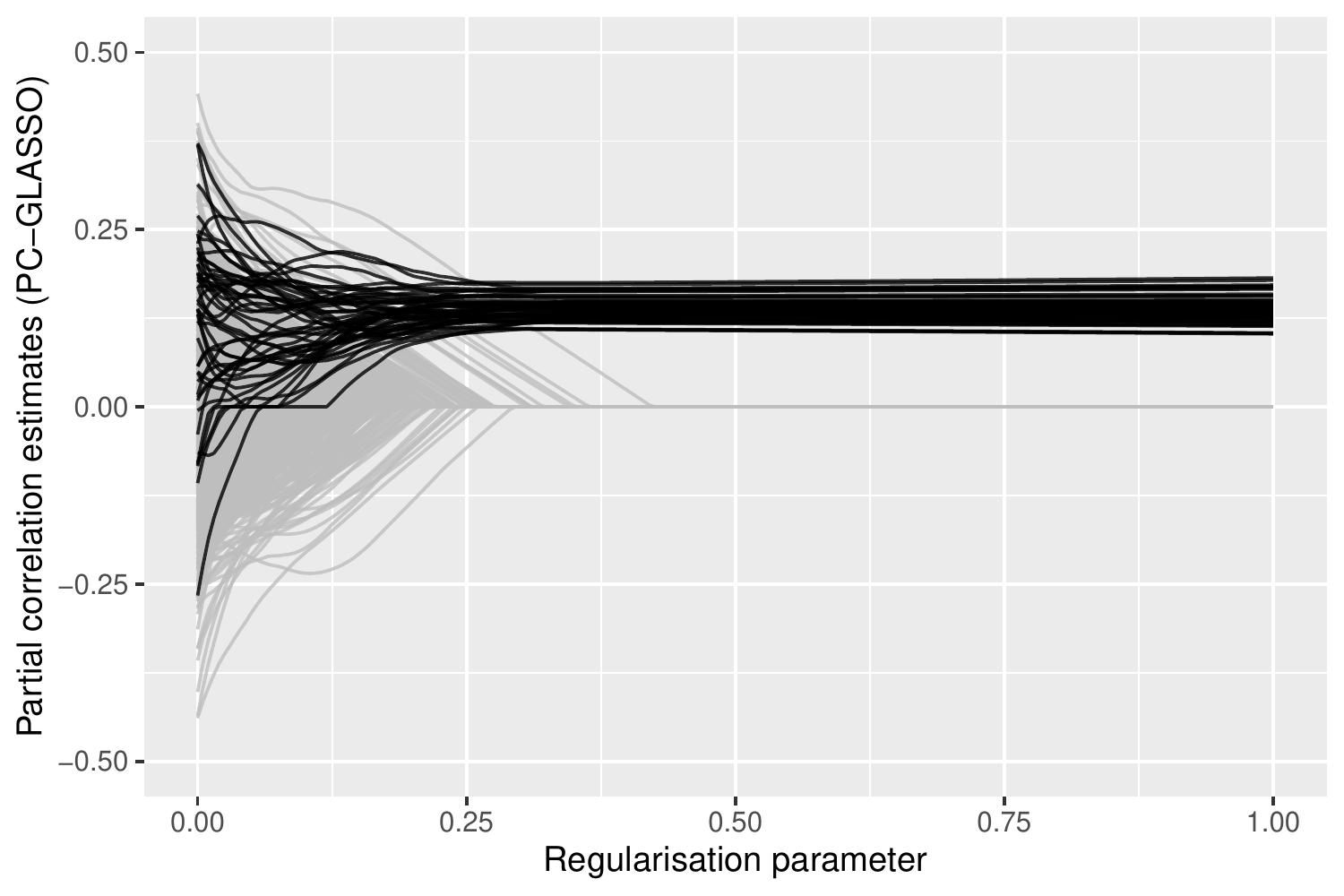}
	\centering
	\includegraphics[scale=0.55]{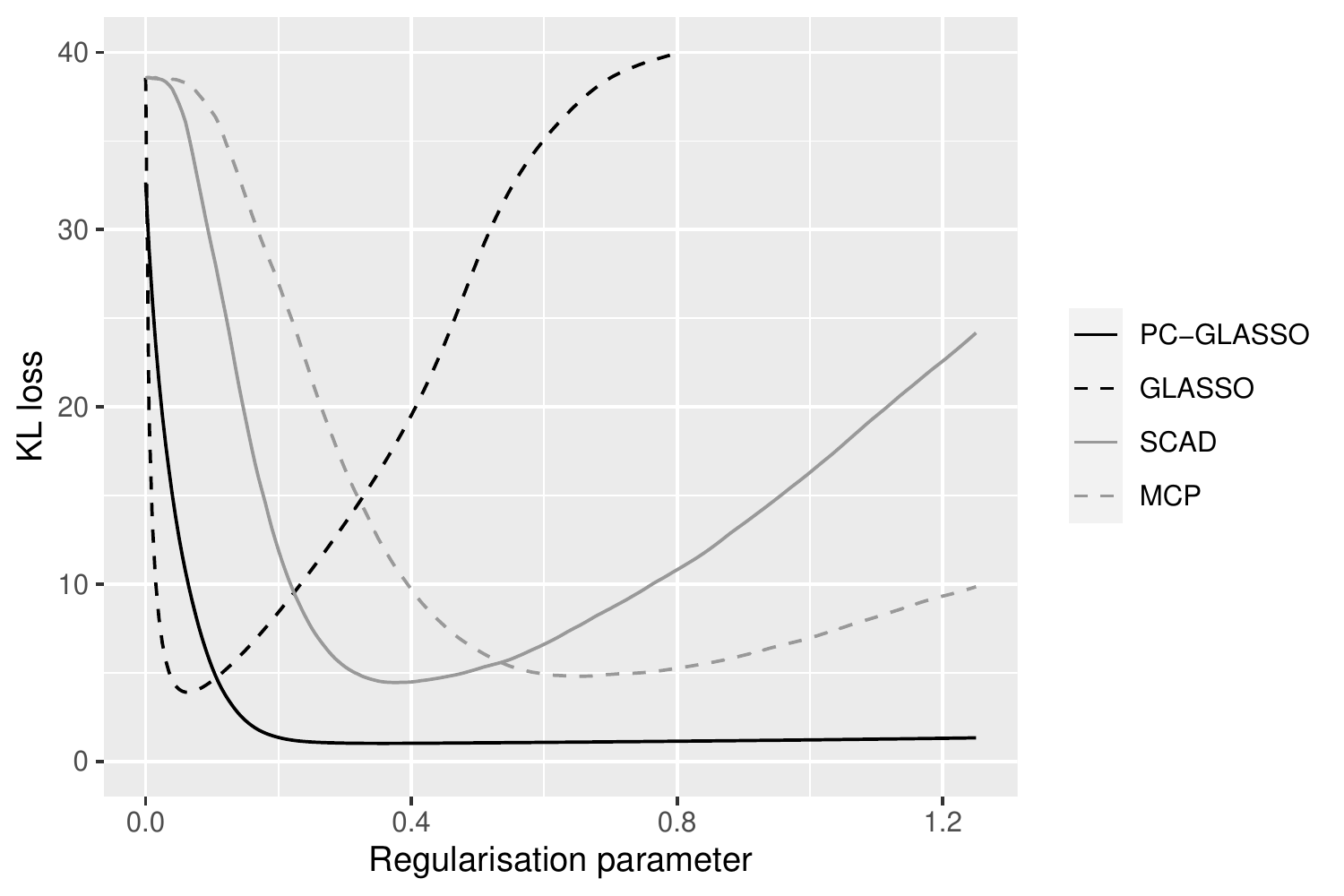}
\caption{Top: partial correlation regularisation paths for GLASSO in the $p=50$ star graph example on the original data (left), and standardised data (right). Estimates of truly non-zero $\theta_{ij}$ are in black.  Bottom: Partial correlation regularisation paths for PC-GLASSO in the $p=50$ star graph example (left) and KL loss over the regularisation paths for different penalties applied to standardised data (right).}
\label{fig:Star1}
\end{figure}

Lack of invariance is not restricted to the GLASSO, but, as we show later, affects essentially all continuous penalties, as well as standard prior distributions in Bayesian settings.  

The paper is organised as follows.
Section \ref{sec:penalised_likelihood} sets notation and reviews popular classes of likelihood penalties
which we refer to as \textit{regular} penalty functions, and their Bayesian equivalents, \textit{regular} prior distributions.
Section \ref{sec:PCGLASSO} introduces a class of penalties and prior distributions on partial correlations, and the PC-GLASSO as a particular case. 
Section \ref{sec:ScaleInv} shows that the PC-GLASSO, as well as the logarithmic and $L_{0}$ penalties are scale invariant, while regular penalty functions are not.
Section \ref{sec:EstEquiv} 
offers an alternative argument for standardising the data when using regular penalties, related to
situations where the likelihood function is exchangeable in two partial correlations, hence one may wish for inference to be exchangeable as well.
Section \ref{sec:Priors} compares the related prior distributions of GLASSO and PC-GLASSO and
Section \ref{sec:Computation} discusses computational issues for the PC-GLASSO and gives a certain conditional convexity result. 
Section \ref{sec:Simulations} shows examples on simulated, gene expression and stock market datasets. 
We end the paper with a short discussion.

\section{Penalised likelihood in Gaussian graphical models}\label{sec:penalised_likelihood}

Let $X = ( X^{(1)}, ..., X^{(p)} ) \sim \mathrm{N}( \mu, \Sigma )$ be a $p$-dimensional multivariate Gaussian random vector with unknown mean $\mu \in \mathbb{R}^p$ and $p \times p$ positive-definite covariance $\Sigma = (\sigma_{ij})_{i \leq i,j \leq p}$.  Suppose we observe $n$ independent samples $( X_{1},\dots, X_{n} )$ of $X$ and denote their sample covariance by $S$.
Our goal is to estimate the precision matrix $\Theta = ( \theta_{ij} )_{ 1 \leq i, j \leq p } = \Sigma^{-1}$. 

A common assumption in Gaussian graphical models is that the data generating process is governed by a sparse undirected graph so that $\Theta$ is a sparse matrix with many zero entries, and we have a particular interest in the location of its zero entries.  This is due to the equivalence between zero partial covariances and conditional independencies in Gaussian graphical models. 
The most common frequentist approach to sparse estimation is to maximise a penalised likelihood function of the form
$
l( \Theta \mid S ) - Pen( \Theta ),
$
where
\begin{align} \label{eq:PenLike}
l( \Theta \mid S ) = \frac{ n }{ 2 } \left[ \log( \det( \Theta ) ) - \mathrm{tr}( S \Theta ) - p \log( 2 \pi ) \right], 
\end{align}
is the log-likelihood function, $Pen( \Theta )$ some penalty function %on $\Theta$ 
and $\mathrm{tr}(A)$ the trace of $A$.
Most popular choices (discussed below) consider penalties
that are additive and monotone in $|\theta_{ij}|$, which we refer to as {\it separable penalties}, and in particular the subclass of penalties differentiable everywhere other than zero, which we refer to as {\it regular penalties}.

\begin{mydef} \label{def:separable_penalty}
A penalty function $Pen(\Theta)$ is \textit{separable} if
$$ Pen( \Theta ) = \sum_{i \leq j} pen_{ij}(  \theta_{ij} ),$$
where 
$pen_{ii}:(0,\infty) \rightarrow \mathbb{R}$ and $pen_{ij}: \mathbb{R} \rightarrow \mathbb{R}$ are non-decreasing in $\theta_{ii}$ and $|\theta_{ij}|$ respectively for all $i$ and $i<j$.

A separable penalty is \textit{regular} if $pen_{ii}=pen_{jj}$ for all $(i,j)$ and, for all $i < j$, $pen_{ij}$ does not depend on $(i,j)$, is symmetric about $0$ and differentiable away from $0$.
\end{mydef}

Most popular penalty functions used for Gaussian graphical models are regular.
The GLASSO is a prominent example
using an $L_1$ penalty to produce the point estimate
\begin{align}\label{eq:GLASSO}
\Theta_{\mathrm{GLASSO}}^{\rho}(S) &= \arg \max \log( \det( \Theta ) ) - \mathrm{tr}( S \Theta ) 
- \rho \sum_{i=1}^p \sum_{j=1}^p \vert \theta_{ij} \vert
\end{align}
for some given regularization parameter $\rho \geq 0$.
See \cite{Meinshausen2006} for an alternative that places $L_1$ penalties on the full conditional regression of each $X^{(i)}$ given $X^{-(i)}$,
\cite{Banerjee2008} for computational methods based on parameterising \eqref{eq:GLASSO} in terms of $\Sigma$ and
\cite{Yuan2007} for a variation that omits the diagonal of $\Theta$ from the penalty. 
Other popular regular penalties include the SCAD penalty \citep{Fan2001,Fan2009}
and the MCP penalty \citep{Zhang2010,Wang2016},
which were proposed to reduce bias in the estimation of large entries in $\Theta$ relative to the $L_1$ penalty. 

Another notable regular penalty is the $L_{0}$ penalty
\begin{align}\label{eq:l0Pen}
Pen(\Theta) = \rho \sum_{i < j} \mathbb{I}(\theta_{ij} \neq 0).
\end{align}

The adaptive LASSO \citep{Zhou2009,Fan2009} is an important example of a non-regular penalty.  It uses an $L_{1}$ penalty where weights depend on the data via some initial estimate of $\Theta$, and hence does not satisfy Definition \ref{def:separable_penalty}.  
However, as noted by \cite{Buhlmann2008} and \cite{Candes2008}, the adaptive LASSO can be seen as a first-order approximation of the logarithmic penalty where $pen_{ij}(\theta_{ij}) = \rho \log( |\theta_{ij}| )$, which is regular.  Both papers propose an iterative version of adaptive LASSO that formally targets this logarithmic penalty.  

There is a well known equivalence between penalised likelihood and maximum a posteriori estimates in Bayesian frameworks.
In particular, the estimate under a penalty $Pen$ is equal to the mode of the posterior distribution under the prior density $\pi(\Theta) \propto \exp(-Pen(\Theta)) \mathbb{I}( \Theta \in \mathcal{S} )$ where $\mathcal{S}$ is the set of symmetric, positive definite matrices.  With this in mind we define \textit{separable} and \textit{regular prior distributions}.

\begin{mydef}\label{def:SepPriors}
A prior distribution with density $\pi$ on $\Theta$ is \textit{separable} if
$$ \pi(\Theta) = \prod_{i \leq j} \pi_{ij}(\theta_{ij}) \mathbb{I}( \Theta \in \mathcal{S} ) $$
where $\pi_{ii}$ is a density function with support $(0,\infty)$ and $\pi_{ij}$ is a density function with support $\mathbb{R}$ which are non-increasing in $\theta_{ii}$ and $|\theta_{ij}|$ respectively for all $i$ and $i \leq j$.

A separable prior distribution is \textit{regular} if $\pi_{ii}=\pi_{jj}$ for all $(i,j)$ and for all $i < j$, $pen_{ij}$ does not depend on $(i,j)$, is symmetric about $0$ and differentiable away from $0$.
\end{mydef}

The correspondence between penalised likelihoods and prior distributions has been utilised by the Bayesian LASSO regression of \cite{Park2008} and \cite{Hans2009} and in Gaussian graphical models by \cite{Wang2012} and \cite{Khondker2013}. Of particular interest to this paper, \cite{Wang2012} showed that under the GLASSO prior the marginal prior distribution of partial correlations does not depend on the regularisation parameter. We explore this further in Section \ref{sec:Priors}.  The Bayesian interpretation has also been used to create new penalties functions, for example by \cite{Banerjee2015} and \cite{Gan2018}, both of whom set mixture priors on the entries of $\Theta$.

\section{Partial Correlation Graphical LASSO}\label{sec:PCGLASSO}

We propose basing penalties on a reparameterisation of $\Theta$ in terms of the (negative) partial correlations
$$\Delta_{ij} := \frac{\theta_{ij} }{ \sqrt{ \theta_{ii} \theta_{jj} } } = -\mathrm{corr}\left( X^{(i)}, X^{(j)} \mid X^{-(ij)} \right).$$
where $X^{-(ij)}$ denotes the vector $X$ after removing $X^{(i)}$ and $X^{(j)}$.

The precision matrix can be decomposed as
$\Theta = \theta^{\frac{1}{2}} \Delta \theta^{\frac{1}{2}}$,
where $\theta= \text{diag}(\Theta)$ and $\Delta$ is the matrix with unit diagonal and off-diagonal entries $\Delta_{ij}$.  
The penalised likelihood function then becomes
\begin{align}\label{eq:PenLikePCs}
\frac{n}{2} \left[ \log( \det( \Delta ) ) + \sum_{i} \log( \theta_{ii} ) - \mathrm{tr}( S \theta^{\frac{1}{2}} \Delta \theta^{\frac{1}{2}} ) \right] - Pen( \theta, \Delta ).
\end{align}

We believe that partial correlations are a better measure of dependence than the off-diagonals $\theta_{ij}$, in that they are easier to interpret and invariant to scalar multiplication of the variables.
We now introduce a class of additive penalties in this parameterisation, a corresponding prior class, and subsequently state our PC-GLASSO as a particular case.

\begin{mydef}

A penalty $Pen$ is {\it partial correlation separable} (PC-separable) if it is of the form
$$ Pen(\theta,\Delta) = \sum_{i} pen_{ii}( \theta_{ii} ) + \sum_{i < j} pen_{ij}( \Delta_{ij} ), $$
where $pen_{ii}:(0,\infty) \rightarrow \mathbb{R}$ and
$pen_{ij}: [-1,1] \rightarrow \mathbb{R}$ are non-decreasing in $\theta_{ii}$ and $|\Delta_{ij}|$ respectively, 
for all $i$ and $i<j$.

A PC-separable penalty function is \textit{symmetric} if $pen_{ii}=pen_{jj}$ for all $(i,j)$ and, for all $i<j$, $pen_{ij}$ does not depend on $(i,j)$ and is symmetric about $0$.
\label{def:pc_separable}
\end{mydef}

Note that Definition \ref{def:pc_separable} includes formulations that do not penalise the diagonal entries, i.e. $pen_{ii}(\theta_{ii})=0$. Note also that
the $L_{0}$ and logarithmic penalties are PC-separable since 
$\theta_{ij}=0$ if and only if $\Delta_{ij}=0$ and $\log(|\theta_{ij}|) = \log( |\Delta_{ij}| ) + \log( \theta_{ii} ) + \log( \theta_{jj} )$.

\begin{mydef}
A prior $\pi( \theta, \Delta )$ is (symmetric) PC-separable if
the penalty function $Pen( \theta, \Delta ) = -\log( \pi( \theta, \Delta ) )$ is (symmetric) PC-separable.
\label{def:pc_separable_prior}
\end{mydef}

Any PC-separable prior can be written as $$ \pi( \theta, \Delta ) \propto \prod_{i} \pi_{ii}( \theta_{ii} ) \prod_{i < j} \pi_{ij}( \Delta_{ij} ) \mathbb{I}( \Delta \in \mathcal{S}_{1} ), $$
where $\mathcal{S}_{1}$ is the set of symmetric, positive definite matrices with unit diagonal.

PC-GLASSO is a symmetric PC-separable penalty applying the $L_1$ norm to the partial correlations
$ pen_{ij}( \Delta_{ij} ) = n \rho | \Delta_{ij} |, $
and a logarithmic penalty to the diagonal $ pen_{ii}( \theta_{ii} ) = 2\log(\theta_{ii}). $
The penalised likelihood function, after removing constants, is given by
\begin{align}\label{eq:PenLikeFn}
\log( \det( \Delta ) ) + \left( 1 - \frac{ 4 }{ n } \right) \sum_{i} \log( \theta_{ii} ) - \mathrm{tr}\left( S \theta^{\frac{1}{2}} \Delta \theta^{\frac{1}{2}} \right) - \rho \sum_{i \neq j} \vert \Delta_{ij} \vert.
\end{align}

The logarithmic penalty on the diagonal entries ensures scale invariance of the PC-GLASSO (Section \ref{sec:ScaleInv}).  A coefficient of 2 is used since in the univariate $p=1$ case this minimises the asymptotic mean squared error of the estimated precision amongst logarithmic penalties (see Appendix \ref{subsec:Appendix-MSE}).
Although many methods use the same penalty forms for diagonal and off-diagonal entries, 
it seems natural to use different forms
since the former do not aim to induce sparsity.
For example, \cite{Yuan2007} argued for a GLASSO framework where one does not penalise the diagonal.  

As usual, one may calculate the PC-GLASSO estimate for a sequence of regularisation parameters $\rho$ and select the solution that maximizes some suitable criterion.  In Section \ref{sec:Simulations} we used the Bayesian information criterion (BIC), which selects the estimate minimising
\begin{align}\label{eq:BIC}
\mathrm{BIC}(\hat{\Theta},S) = \log(n) \sum_{i < j} \mathbb{I}( \hat{\theta}_{ij} \neq 0 ) - 2l(\hat{\Theta} \mid S),
\end{align}
Parameter selection via the BIC has been shown to provide consistent graphical model selection when used with the SCAD and MCP penalties.  Other potential criteria that have been explored for GLASSO are cross validation and the extended Bayesian information criterion (EBIC, \cite{Foygel2010}), which we also consider in our real data applications.  For further discussion see, for example, \cite{Vujacic2015}.

There are some examples of penalty functions for Gaussian graphical models based on partial correlations.  \cite{Ha2014} utilised a ridge penalty. The space method of \cite{Peng2009}, similarly to PC-GLASSO, uses an $L_{1}$ penalty on the partial correlations, but in combination with a function other than the log-likelihood.  \cite{Azose2018} introduced a separable prior on the \textit{marginal} correlations. They argued that a key benefit of their prior is the ability to specify beliefs about correlations. A similar argument can be made for PC-separable priors allowing one to specify prior beliefs on partial correlations.

\section{Scale invariance}\label{sec:ScaleInv}

A key property of graphical models is invariance to scalar multiplication.
In the Gaussian case, if we consider the transformation $DX$ for some fixed diagonal $p \times p$ matrix $D$ with non-zero diagonal, then $DX$ is also Gaussian with precision matrix
\begin{align}\label{eq:ScaledTheta}
\Theta_{D} = D^{-1} \Theta D^{-1}.
\end{align}
In particular, the zero entries of $\Theta_{D}$ are identical to those of $\Theta$.

We argue that it is desirable for an estimator of $\Theta$ to mirror the relationship in (\ref{eq:ScaledTheta}) under scalar multiplication of the data, a property we call \textit{scale invariance}.
We now show that, among regular penalty functions, only the $L_{0}$ and logarithmic penalties are scale invariant, whereas PC-separable penalties are.
Recall that any estimator can be made scale invariant by standardising the data to unit sample variances prior to obtaining the estimate, but as discussed this has an effect on inference. 
We start by defining two notions of scale invariance related to the point estimate and to the recovered graphical structure.

\vspace{10pt}
\begin{mydef}\label{def:scale_invariant}

An estimator $\hat{\Theta}$ is \textit{scale invariant} if for any sample covariance matrix $S$ and any diagonal $p \times p$ matrix $D$ with non-zero diagonal entries, $$\hat{\Theta}(DSD) = D^{-1} \hat{\Theta}(S) D^{-1}.$$ 
$\hat{\Theta}$ is \textit{selection scale invariant} if $\hat{\Theta}(S)$ and $\hat{\Theta}(DSD)$ have identical zero entries for any $S$ and $D$.

\end{mydef}

Scale invariance ensures that the estimate under the scaled data
corresponds to
that under the original data as in (\ref{eq:ScaledTheta}). Meanwhile selection scale invariance ensures that one recovers the same graphical structure under scalar multiplications.
It is clear that scale invariance implies selection scale invariance.

We now present results on the scale invariance of different penalties. Note that the results could equivalently be written in terms of the maximum a posteriori estimate under corresponding prior distributions.  All proofs are in Appendix \ref{subsec:Appendix-ScaleInvProofs}.

\begin{myprop}\label{prop:RegularScaleInv}

Let $\hat{\Theta}$ be an estimator based on a regular penalty, and suppose that there exists a sample covariance matrix $S$ such that $\hat{\Theta}(S)$ is not a diagonal matrix.  Then $\hat{\Theta}$ is scale invariant if and only if $pen_{ij}$ is either an $L_{0}$ or logarithmic penalty, and $pen_{ii}$ is either a constant or a logarithmic penalty.

\end{myprop}

In particular, the GLASSO, SCAD and MCP estimators are not scale invariant.
Further, as illustrated in Figure \ref{fig:Star1} these estimators are also not selection scale invariant. We conjecture that lack of selection scale invariance holds more widely for regular penalty functions, but settle with the counterexample for these three cases provided by Figure \ref{fig:Star1}.

We present an example to further illustrate how scaling can affect the inferred conditional independence structure.
Suppose we observe the inverse sample covariance matrix
$$ S^{-1} = \begin{pmatrix}
1 & 0.5 & 0 \\
0.5 & 1 & 0.25 \\
0 & 0.25 & 1
\end{pmatrix} $$
The left panel in Figure \ref{fig:ScaleInv} shows the associated GLASSO estimates $\Theta_{\mathrm{GLASSO}}^{\rho}(S)$.
The right panel considers the situation where the data were given on a different scale, specifically the sample covariance 
is $DSD$ where $D$ has diagonal entries 1, 1 and 10, 
and provides the estimates $D\Theta_{\mathrm{GLASSO}}^{\rho}(DSD)D$. 
The estimates set to zero, as well as their relative magnitudes, differ significantly depending on the scale of the data.
We observed similar results for the SCAD and MCP penalties (not shown, for brevity).

\begin{figure}[ht]
\vspace{10pt}
	\centering
	\includegraphics[scale=0.7]{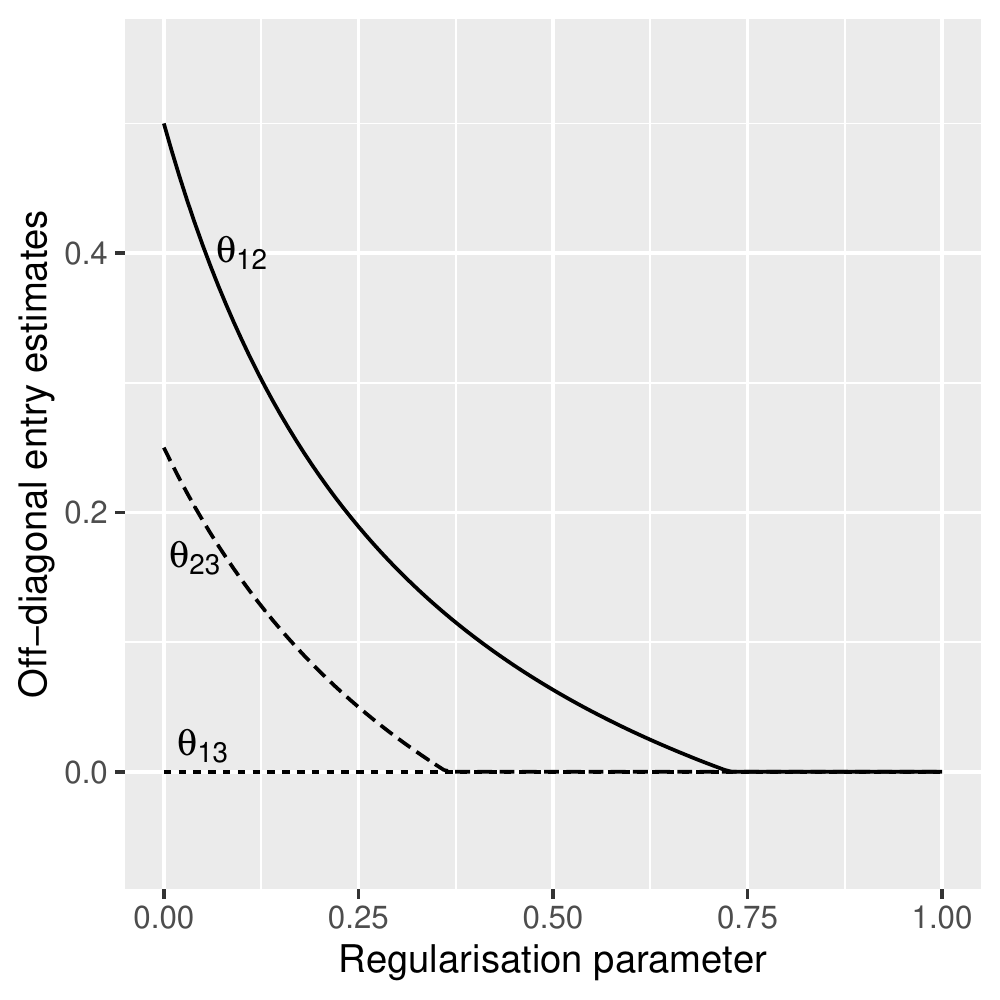}
	\centering
	\includegraphics[scale=0.7]{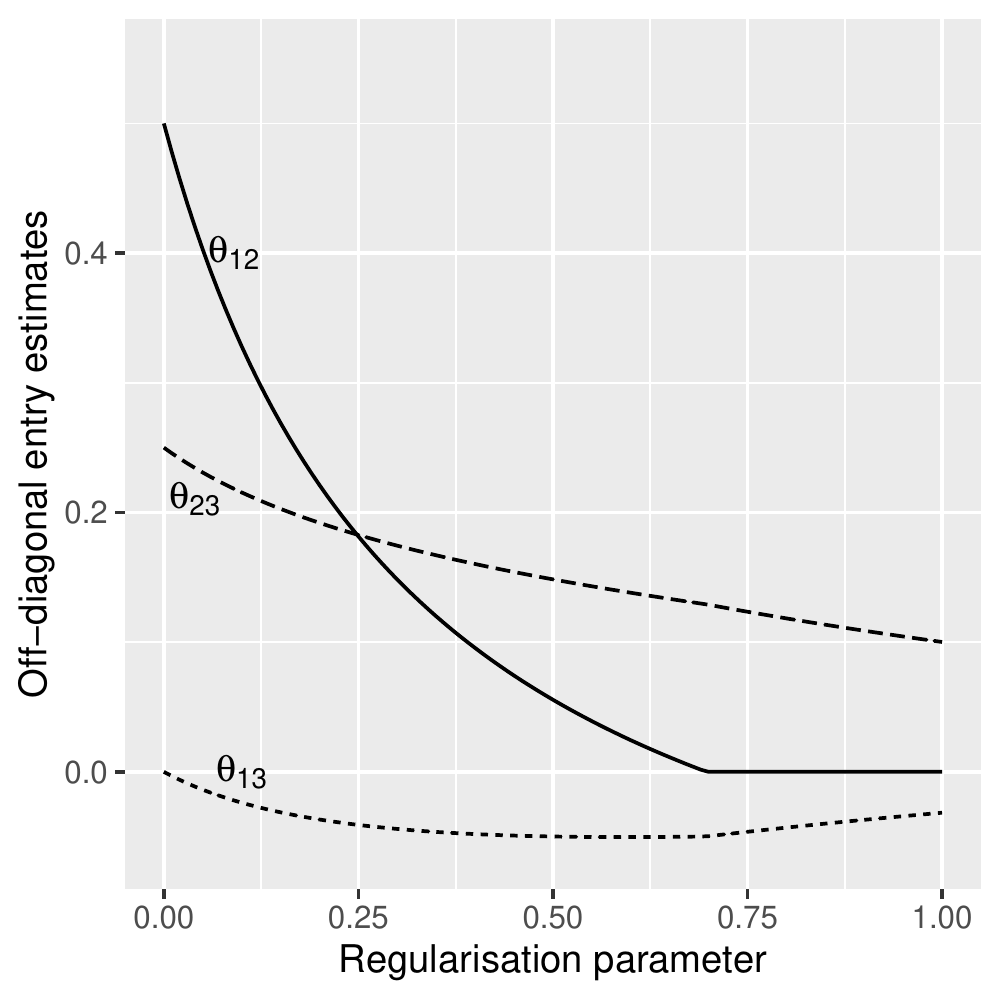}
\caption{Estimated off-diagonal entries $\Theta_{\mathrm{GLASSO}}^{\rho}(S)$ (left) and $D\Theta_{\mathrm{GLASSO}}^{\rho}(DSD)D$ (right) for regularisation parameter $\rho \in [ 0, 1 ]$.}
\label{fig:ScaleInv}
\end{figure}

As shown in Proposition \ref{prop:RegularScaleInv}, the only scale invariant regular penalties are the $L_{0}$ and logarithmic penalties, both of which are also PC-separable.
In fact scale invariance holds more widely in PC-separable penalties, from which it follows that PC-GLASSO is scale invariant.

\begin{myprop}\label{Prop:PCSepScaleInv}

Any %penalised likelihood 
estimator based on a symmetric PC-separable penalty is scale invariant, provided %$pen_{ii}(\theta_{ii})=0$ or  
$pen_{ii}(\theta_{ii}) = c \log(|\theta_{ii}|)$ for some constant $c \geq 0$.

\end{myprop}

In the Bayesian framework, Proposition \ref{Prop:PCSepScaleInv} implies scale invariance of the a posteriori mode under symmetric PC-separable priors.  That is, let $\tilde{\Theta}=\hat{\Theta}(DSD)$ be the posterior mode under the scaled sample covariance, then the mode under the original sample covariance is $\hat{\Theta}(S)= D \tilde{\Theta} D$.  Hence, the maxima of the two posterior densities are $\pi\left( \tilde{\Theta} \mid DSD \right)$ and $\pi\left( D \tilde{\Theta} D \mid S \right)$.

In fact a stronger property holds for the entire posterior distribution, that PC-separable priors lead to scale-invariant posterior inference, as defined below. % is maintained under scalar multiplication and thus all posterior inference is scale invariant.

\begin{mydef}\label{def:PriorScaleInv}

Let $\pi(\Theta)$ be a prior density, $S$ a sample covariance and $D$ a diagonal matrix with non-zero diagonal.  Let the posterior density associated to $S$ be $\pi( \Theta \mid S ) \propto L(\Theta \mid S) \pi(\Theta)$, and that associated to $DSD$ be $\pi(\Theta \mid DSD) \propto L(\Theta \mid DSD) \pi(\Theta)$ where $L$ is the Gaussian likelihood function.

$\pi(\Theta)$ leads to scale-invariant posterior inference if for any $(S,D)$
\begin{align}\label{eq:PriorScaleInv}
\mathbb{P}_{\pi}\left( \Theta \in A \mid DSD \right) = \mathbb{P}_{\pi}\left( \Theta \in  A_{D} \mid S \right)
\end{align}
for all measurable sets $A$ where $A_{D} = \{ \Theta : D^{-1} \Theta D^{-1} \in A \}$.

\end{mydef}

In particular, (\ref{eq:PriorScaleInv}) implies that the two posterior distributions on the partial correlations $\Delta$ are equal up to appropriate sign changes i.e. when $D$ has all positive entries, $\pi(\Delta \mid S)= \pi(\Delta \mid DSD)$ (since $\Delta$ associated to $\Theta$ is equal to that associated to $D \Theta D$).

\begin{myprop}\label{prop:PriorScaleInv}
Any symmetric PC-separable prior distribution with $\pi_{ii}(\theta_{ii}) \propto \theta_{ii}^{-c}$ for some constant $c \geq 0$ leads to scale-invariant posterior inference.
\end{myprop}

\section{Exchangeable inference}\label{sec:EstEquiv}

We now discuss an alternative view on the desirability of standardising the data when using regular penalties, based on notions of exchangeable inference. The simplest situation occurs when the likelihood function is exchangeable in two or more $\Delta_{ij}$'s, for example when two rows in the sample correlation matrix $R= \mbox{diag}(S)^{-1/2} S \mbox{diag}(S)^{-1/2}$ are equal (up to the necessary index permutations).
In such a situation the likelihood provides the same information on these $\Delta_{ij}$'s, hence it seems desirable to obtain the same inference for all of them. 
If the log-likelihood is exchangeable in some parameters, then any symmetric PC-separable penalty and prior trivially leads to exchangeable inference on those parameters.
Yet, as illustrated in our example below, regular penalties can lead to significantly different inference (unless one standardises the data).

\begin{example}

Consider a $p=4$ setting where the data-generating truth follows a star graph, featuring an edge between $X^{(1)}$ and each of $X^{(2)},X^{(3)},X^{(4)}$, and no other edges.
Specifically, suppose that truly $\theta_{11}=\theta_{22}=\theta_{44}=1$, $\theta_{33}=4$, $\theta_{12}=\theta_{14}=-0.5$ and $\theta_{13}=-1$,
so that the data-generating partial correlations are $\Delta_{12}=\Delta_{13}=\Delta_{14}=0.5$, and $\Delta_{ij}=0$ for all remaining $(i,j)$.
Consider an ideal scenario where the sample covariance $S$ matches the data-generating truth. That is,
$$ S^{-1} = \begin{pmatrix}
1 & -0.5 & -1 & -0.5 \\
-0.5 & 1 & 0 & 0 \\
-1 & 0 & 4 & 0 \\
-0.5 & 0 & 0 & 1
\end{pmatrix};
\hspace{3mm}
S= \begin{pmatrix}
 4 &  2 & 1 &  2 \\
 2 &  2 & 0.5 &  1 \\
 1 &  0.5 & 0.5 &  0.5  \\
 2 &  1 & 0.5 &  2  \\
\end{pmatrix};
\hspace{3mm}
R= \begin{pmatrix}
 1 & 1/\sqrt{2} & 1/\sqrt{2} & 1/\sqrt{2} \\
 1/\sqrt{2} & 1   & 0.5  & 0.5 \\
 1/\sqrt{2} & 0.5 & 1    & 0.5 \\
 1/\sqrt{2} & 0.5 & 0.5  &  1 \\
\end{pmatrix}
$$
In this example, the likelihood is exchangeable in $(\Delta_{12},\Delta_{13},\Delta_{14})$, hence it seems desirable that $\hat{\Delta}_{12}=\hat{\Delta}_{13}=\hat{\Delta}_{14}$.  The estimates for the remaining $\Delta_{ij}$ should ideally be close to 0, their true value.

The left panel of Figure \ref{fig:EstEq} shows the GLASSO path for the partial correlations.
The estimate for $\Delta_{13}$ is fairly different than for $\Delta_{12}$ and $\Delta_{14}$, and so is the range of $\rho$'s for which they are set to 0.
Note however that the estimates for the remaining $\Delta_{ij}$'s are close to 0.
To address this issue, one may note that the diagonal of $S$ is not equal to 1. Indeed, if one standardises the data, so that the sample covariance is equal to $R$, one obtains
the center panel of Figure \ref{fig:EstEq}.
Now $\hat{\Delta}_{12}= \hat{\Delta}_{13}=\hat{\Delta}_{14}$ for any regularisation parameter $\rho$, as we argued is desirable.
However, the estimates for truly zero parameters are somewhat magnified for $\rho \in [0.05,0.35]$.

The PC-GLASSO estimates (on either the original or standardised data, due to scale invariance) in the right panel of Figure \ref{fig:EstEq} satisfy $\hat{\Delta}_{12}=\hat{\Delta}_{13}=\hat{\Delta}_{14}$, and the truly zero parameters are clearly distinguished.

\begin{figure}[ht]
\vspace{10pt}
	\centering
	\includegraphics[scale=0.45]{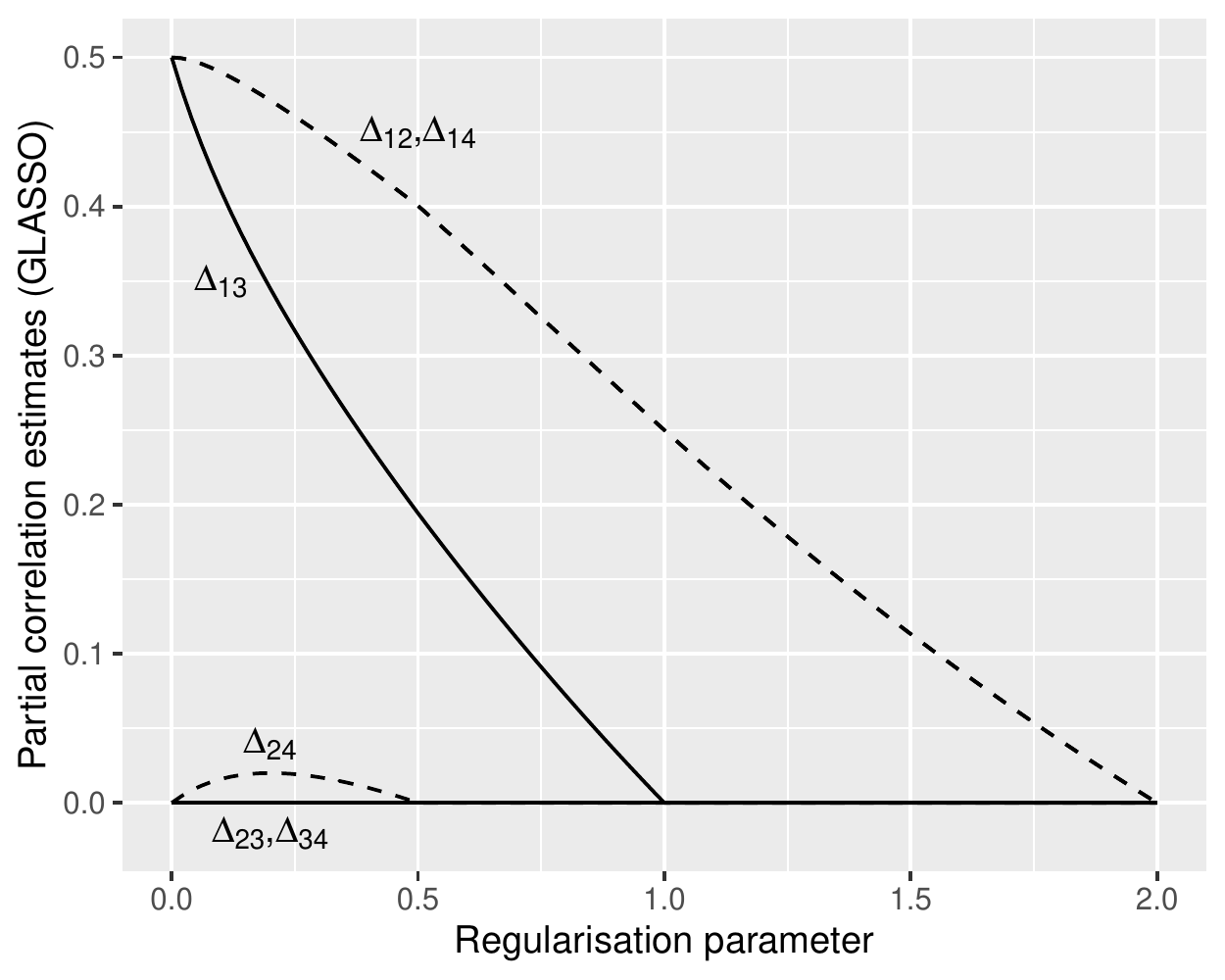}
	\centering
	\includegraphics[scale=0.45]{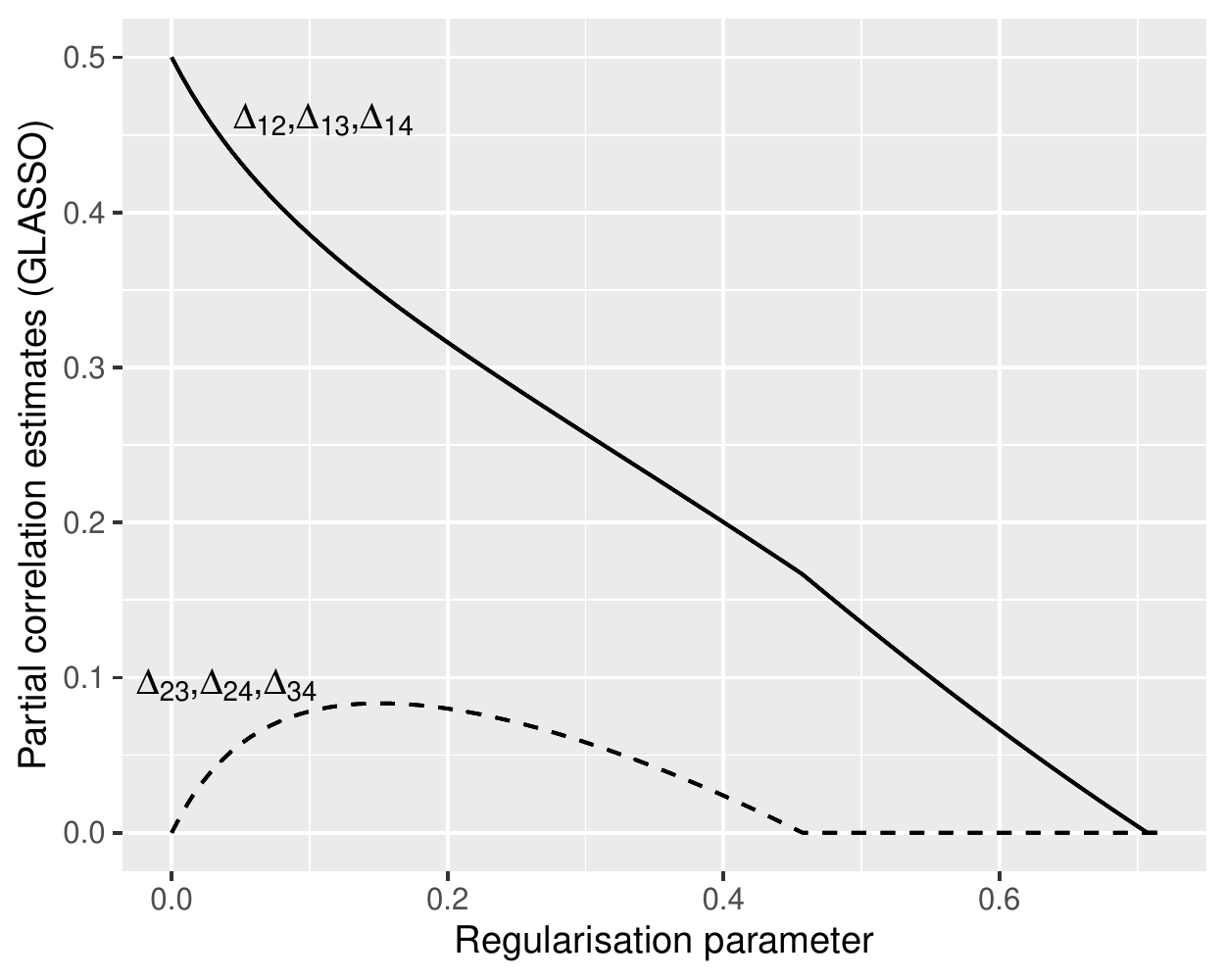}
	\centering
	\includegraphics[scale=0.45]{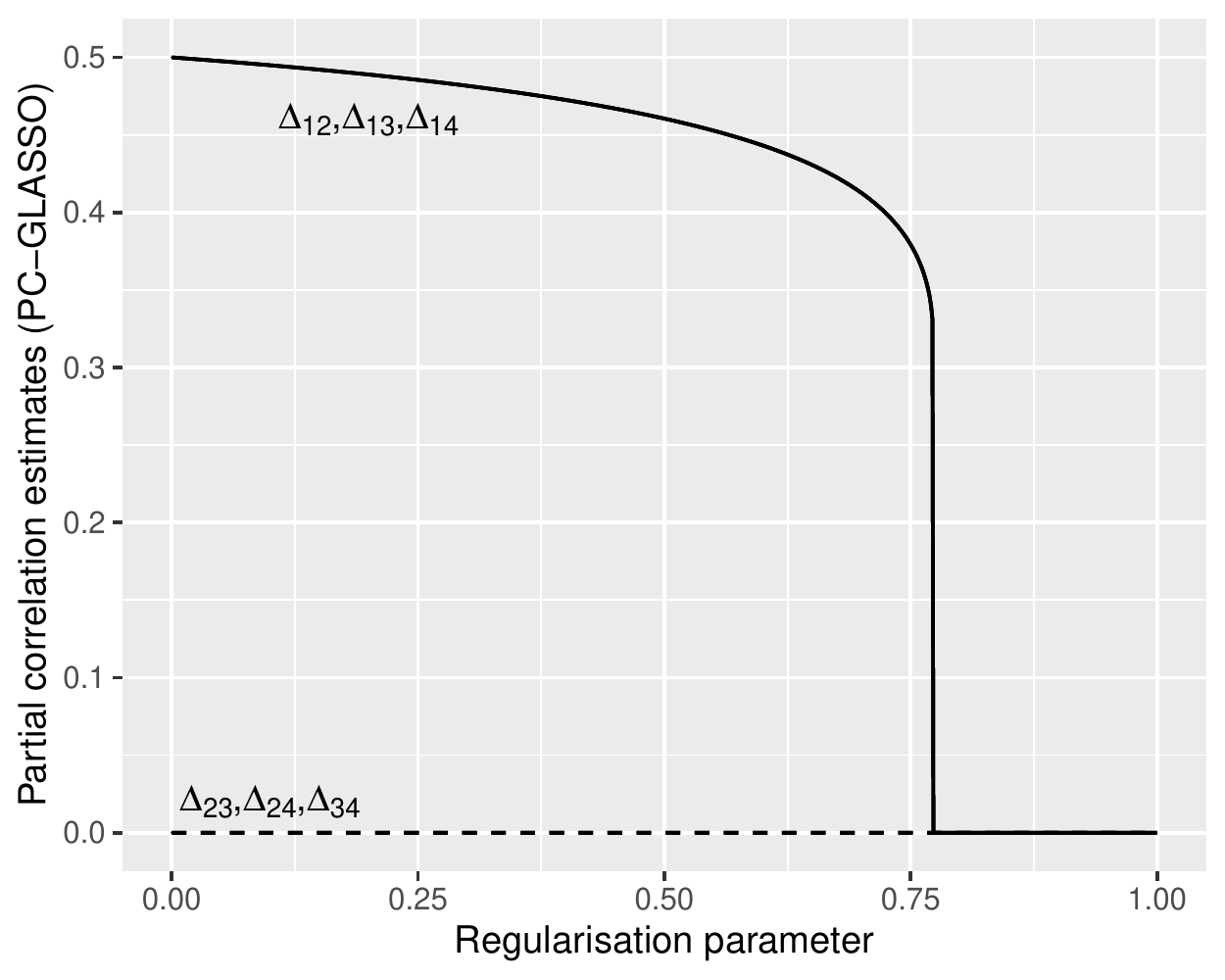}
\caption{Partial correlation regularisation paths in $p=4$ star graph example for GLASSO on the original $S$ (left), standardised $S$ (center) and PC-GLASSO (right).}
\label{fig:EstEq}
\end{figure}
\end{example}

We remark that the notion can be extended to conditional exchangeability, i.e. the likelihood being symmetric in $(\Delta_{ij},\Delta_{kl})$ given the remaining parameters in $\Delta$ and $\theta$.
For example, the likelihood is conditionally exchangeable in $(\Delta_{ij},\Delta_{ik})$ when the sample covariances and precisions are related by the same constant, i.e.
$S_{i j} = c S_{i k} $ and 
$\theta_{kk}^{1/2} = c \theta_{jj}^{1/2}$ for some $c>0$, 
and the partial correlations with other variables are equal, i.e. $\Delta_{j l} = \Delta_{k l}$ for all $l \not\in \{i,j,k\}$.
See Appendix \ref{subsec:Appendix-EstEqu} for additional information and supplementary results.
Conditional exchangeability would be relevant in situations where two variables $(j,k)$ have the same estimated partial correlations with all other variables (e.g. zero), as well as the same sample covariances with a third variable $i$. In such situations, one may wish for equal inference, in particular equal point estimates $\hat{\Delta}_{ij}=\hat{\Delta}_{ik}$.

\section{GLASSO and PC-GLASSO prior distributions}\label{sec:Priors}

In this section we provide further insights into the shrinkage induced by GLASSO and PCGLASSO, by comparing their implied prior distributions in a Bayesian framework.

The GLASSO prior \citep{Wang2012} can be written as
$$ \pi_{G}( \Theta ) \propto \prod_{i} \mathrm{Exp}( \theta_{ii} ; \lambda/2 ) \prod_{i<j} \mathrm{Laplace}( \theta_{ij} ; 0, \lambda^{-1} ) \mathbb{I}(\Theta \in \mathcal{S} ), $$
where $\lambda = n \rho$, 
whereas 
the PC-GLASSO prior
is given by
$$ \pi_{PC}( \theta, \Delta ) \propto \prod_{i} \theta_{ii}^{-2} \prod_{i<j} \mathrm{Laplace}( \Delta_{ij} ; 0, \lambda^{-1} ) \mathbb{I}(\Delta \in \mathcal{S}_{1} ). $$

To illustrate the effect of increasing the parameter $\lambda$ for fixed $p$ (see \cite{Wang2012} for results on growing $p$ for fixed $\lambda$), we sampled from each prior via rejection sampling for $\lambda = 1, 2$ and $4$. Figure \ref{fig:Priors} plots the densities of $\Delta_{12}$ and $\theta_{11}$.
The top left panel verifies the claim of \cite{Wang2012} that the GLASSO prior $\pi_G(\Delta_{ij})$ does not depend on $\lambda$, whereas the bottom panel shows that $\pi_G(\theta_{ii})$ is shrunk towards $0$ as $\lambda$ increases.  
In contrast, the PC-GLASSO prior (top-right panel) on partial correlations $\pi_{PG}(\Delta_{ij})$ concentrates around zero as $\lambda$ grows. The marginals on the diagonal entries are given by $\pi_{PG}(\theta_{ii})\propto\theta_{ii}^{-2}$ regardless of $\lambda$.

This demonstrates a fundamental difference in how GLASSO and PCGLASSO induce sparsity in the $\theta_{ij} = \Delta_{ij} \sqrt{\theta_{ii}\theta_{jj}}$.  PCGLASSO achieves sparsity through regularisation of the partial correlations, while GLASSO does so by shrinking the diagonal $\theta_{ii}$.

\begin{figure}[ht]
\vspace{10pt}
	\centering
	\includegraphics[scale=1]{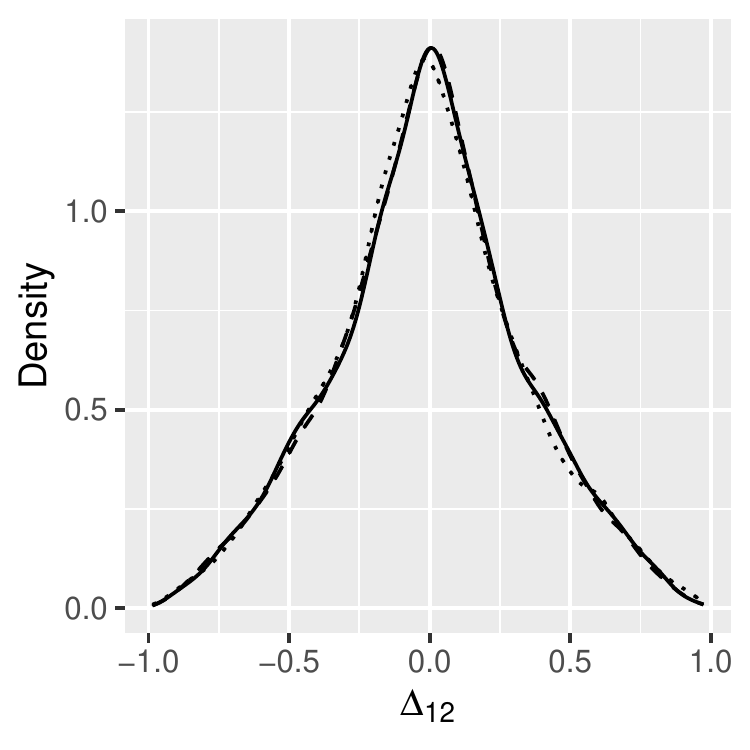}
	\centering
	\includegraphics[scale=1]{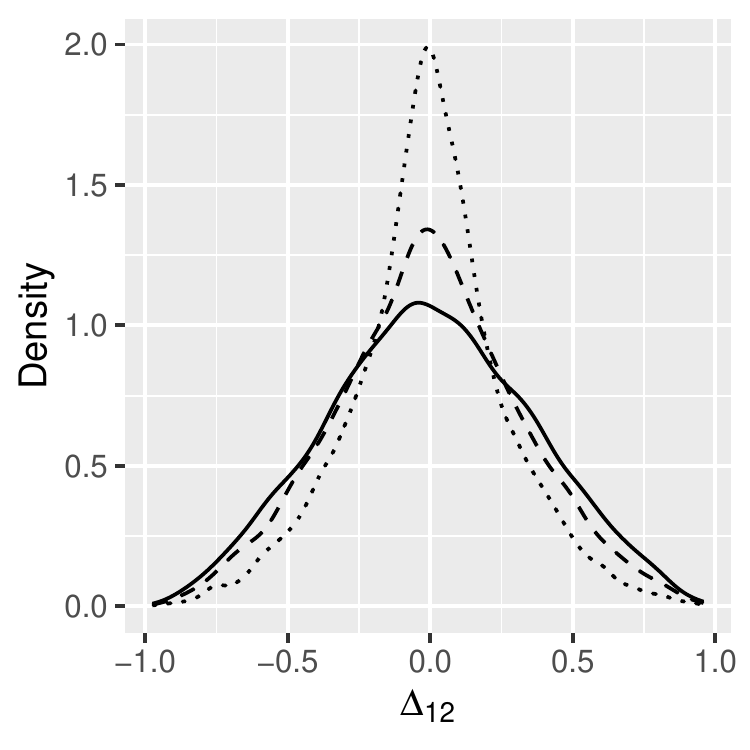}
	\centering
	\includegraphics[scale=1]{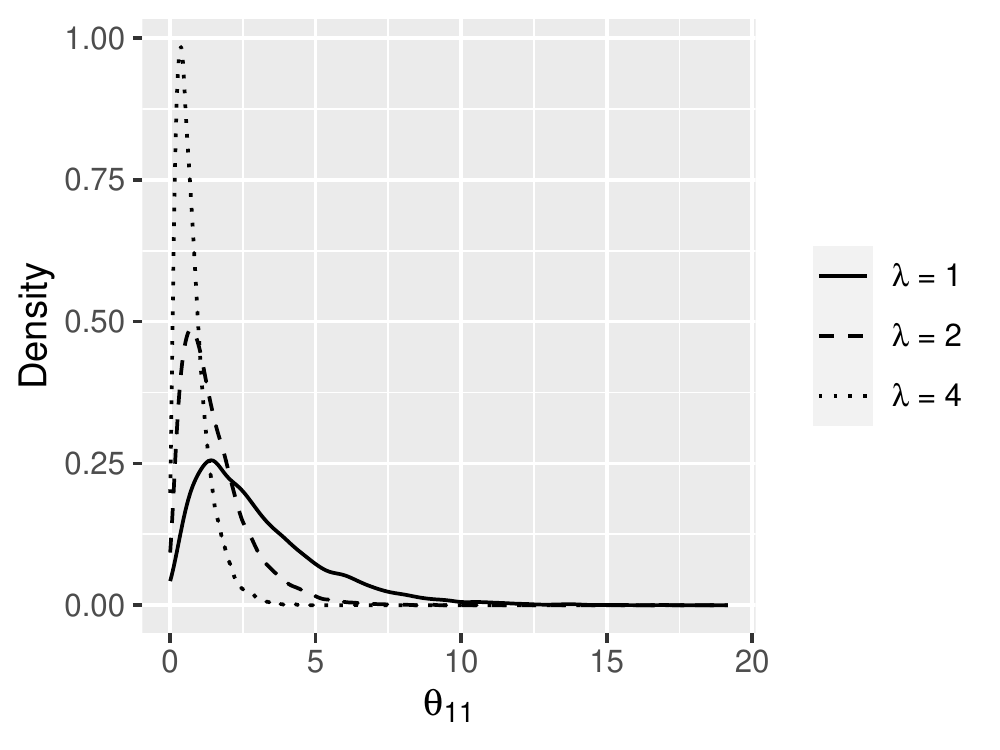}
\caption{Marginal prior densities for the partial correlations under GLASSO prior (top left) and PC-GLASSO prior (top right) and for the diagonal entries under the GLASSO prior (bottom).}
\label{fig:Priors}
\end{figure}

\section{Computation}\label{sec:Computation}

An important feature of GLASSO is its defining of a convex problem that significantly facilitates computation and its theoretical study.
For example, %the R package \textbf{gLASSO} implements the method of 
\cite{Friedman2008} related GLASSO to a sequence of LASSO problems, see also \cite{Sustik2012} for improved algorithms. %implemented in the R package \textbf{gLASSOFast}. 
Computation for non-convex penalties such as SCAD and MCP poses a harder challenge, but the Local Linear Approximation of \cite{Zou2008} greatly facilitates this task, see also \cite{Fan2009}.
The PC-GLASSO optimisation problem is non-convex, however it is conditionally convex given $\theta= \mbox{diag}(\Theta)$.

\begin{myprop}\label{prop:Conv1}

The penalised likelihood function (\ref{eq:PenLikeFn}) is concave in $\Delta$, for any fixed value of $\theta$.  

\end{myprop}

Proposition \ref{prop:Conv1} (proof in Appendix \ref{subsec:Appendix-CompProofs}) opens the possibility to consider block-optimization algorithms, where $\hat{\theta}$ and $\hat{\Delta}$ are updated sequentially, to facilitate computation. In our examples, we took an even simpler strategy and used a coordinate descent algorithm. 
Despite its conceptual simplicity, the algorithm requires careful updating of each parameter to ensure positive definiteness of $\hat{\Delta}$. For brevity we defer details to Appendix \ref{sec:Appendix-CoordDesc} and Algorithms S\ref{alg:RegPath}-S\ref{alg:CoordDesc}.
For the scale of problems addressed in this paper, provided the starting point is close to the optimum then the algorithm typically converges in a few iterations. 
To exploit this observation, when considering a sequence of penalty parameters $0=\rho_{0}<\rho_{1}<\dots<\rho_{k}$, we used the estimated $(\hat{\theta},\hat{\Delta})$ associated to $\rho_{i}$ as the starting point for the problem associated to $\rho_{i+1}$. For $\rho_0=0$ the algorithm is initialised at $S^{-1}$, or at $( S + \alpha I )^{-1}$ where $I$ is the identity matrix if $n<p$.  The matrix $S + \alpha I$ is guaranteed to be invertible and positive definite for any $\alpha$.

\section{Applications}\label{sec:Simulations}

We now assess the performance of PC-GLASSO against GLASSO, SCAD and MCP, setting the regularization parameters via the BIC in (\ref{eq:BIC}). SCAD and MCP have an additional regularization parameter, which we set to the default proposed in \cite{Fan2001} and \cite{Zhang2010} respectively. For all methods we standardised data to unit sample variances, and rescaled the estimates via \eqref{eq:ScaledTheta}.  GLASSO was implemented using the R package \textbf{glasso} and SCAD and MCP using the package \textbf{GGMncv} (see \cite{Williams2020}).

Our primary interest is studying PC-GLASSO versus GLASSO, as they are directly comparable in the sense of using the same $L_1$ penalty structure. We consider SCAD and MCP as benchmarks designed to ameliorate the estimation bias associated to the $L_1$ penalty. Although not considered here for brevity, it would also be interesting to study the use of SCAD and MCP penalties on partial correlations.

\subsection{Simulations}

We considered four simulation scenarios with Gaussian data, truly zero mean and precision matrix $\Theta$ with unit diagonal and off-diagonal entries as follows.

\begin{enumerate}[align=left]

\item[Scenario 1: Star graph -] $ \theta_{ij} = \begin{cases}
-\frac{1}{\sqrt{p}}, & i=1 \text{ or } j=1 \\
0, & \text{otherwise}
\end{cases} $

\item[Scenario 2: Hub graph -] Partition variables into 4 groups of equal size, with each group associated to a `hub' variable $i$.  For any $j\neq i$ in the same group as $i$ we set $\theta_{ij} = \theta_{ji} = \frac{-2}{\sqrt{p}}$ and otherwise $\theta_{ij}=0$.

\item[Scenario 3: AR2 model -] $ \theta_{ij} = \begin{cases}
\frac{1}{2}, & j = i-1,i+1 \\
\frac{1}{4}, & j = i-2,i+2 \\
0, & \text{otherwise}
\end{cases} $

\item[Scenario 4: Random graph -] randomly select $\frac{3}{2}p$ of the $\theta_{ij}$ and set their values to be uniform on $[-1,-0.4] \cup [0.4,1]$, and the remaining $\theta_{ij}=0$. Calculate the sum of absolute values of off-diagonal entries for each column.  Divide each off-diagonal entry by 1.1 times the corresponding column sum and average this rescaled matrix with its transpose to obtain a symmetric, positive definite matrix.
\end{enumerate}

For each setting we used $p=20$ variables, considered sample sizes $n \in \{30, 100\}$ and we performed 100 independent simulations.
To assess estimation accuracy we used the 
Kullback–Leibler (KL) loss
$$ \mathrm{KL}(\Theta,\hat{\Theta}) = -\log(\hat{\Theta}) + \mathrm{tr}(\hat{\Theta} \Theta^{-1}) + \log(\Theta) - p. $$  
To assess model selection accuracy we considered the 
Matthews correlation coefficient (MCC)
$$ \mathrm{MCC}=\frac{ \mathrm{TP}\times\mathrm{TN} - \mathrm{FP}\times\mathrm{FN} }{ \sqrt{ (\mathrm{TP}+\mathrm{FP}) (\mathrm{TP}+\mathrm{FN}) (\mathrm{TN}+\mathrm{FP}) (\mathrm{TN}+\mathrm{FN}) } }, $$
where TP, TN, FP and FN stand for the number of true positives, true negatives, false positives and false negatives (respectively) and measure the ability to recover the true edges in the graph corresponding to $\Theta$. The MCC combines specificity and sensitivity into a single assessment and ranges between $-1$ and $1$, where $1$ indicates perfect model selection.  More information on the MCC can be found in, for example, \cite{Chicco2020}.

Figure \ref{fig:allgraphs} summarises the results. More detailed results, including Frobenius norm, sensitivity and specificity, are in Appendix \ref{sec:Appendix-Results}.
PCGLASSO generally outperformed GLASSO in all scenarios, and either outperformed or was competitive to SCAD and MCP.
More specifically, PCGLASSO strongly outperformed other methods in the Star graph setting in estimation and model selection.
The Star graph is an example where there is a large range in the node degrees, suggesting that penalising partial correlations can be particularly beneficial in such situations.
The AR2 model is the opposite situation where every node has either 1 or 2 edges. Here PCGLASSO still improved significantly over GLASSO, and to a lesser extent over SCAD or MCP in the $n=30$ case, but for $n=100$ the latter two provided better estimation and model selection recovery.
PCGLASSO was also generally better in the Hub and Random graph settings, particularly for $n=30$, although SCAD and MCP offered slight improvements for $n=100$.

Figure \ref{fig:allgraphs_selection} shows the proportion of the 100 simulations in which each edge was selected, illustrating that PCGLASSO generally selected sparser models than GLASSO, particularly in the Star and Hub scenarios.

\begin{figure}[p]
\centering
\begin{tabular}{cc}
\multicolumn{2}{c}{Star graph} \\
\includegraphics[scale=0.6]{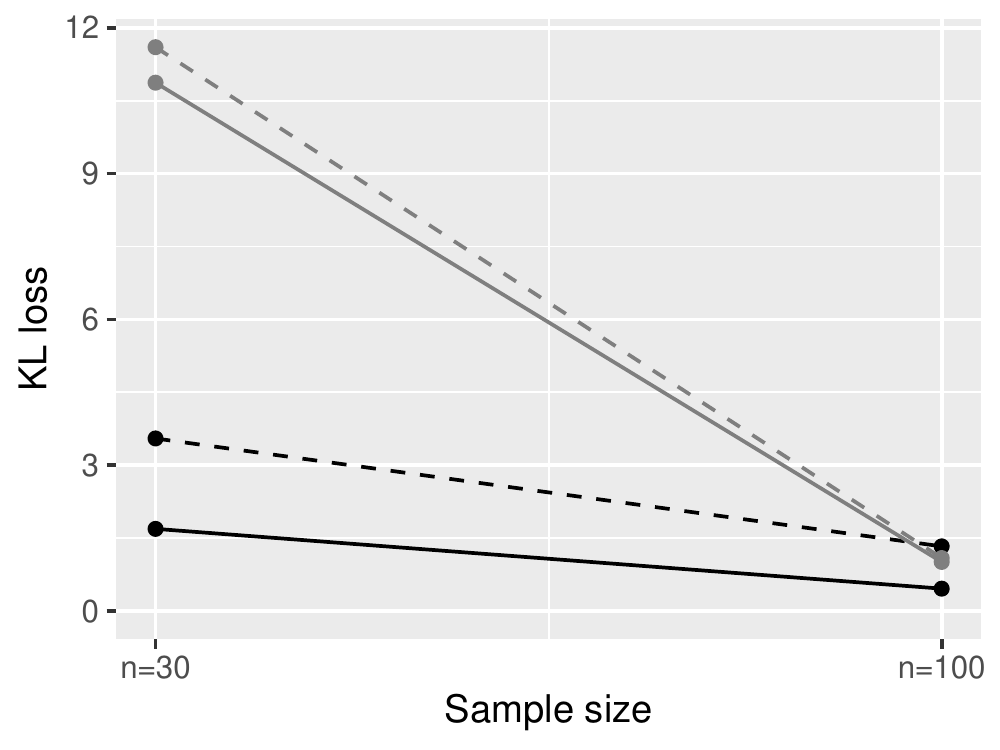} &
	\includegraphics[scale=0.6]{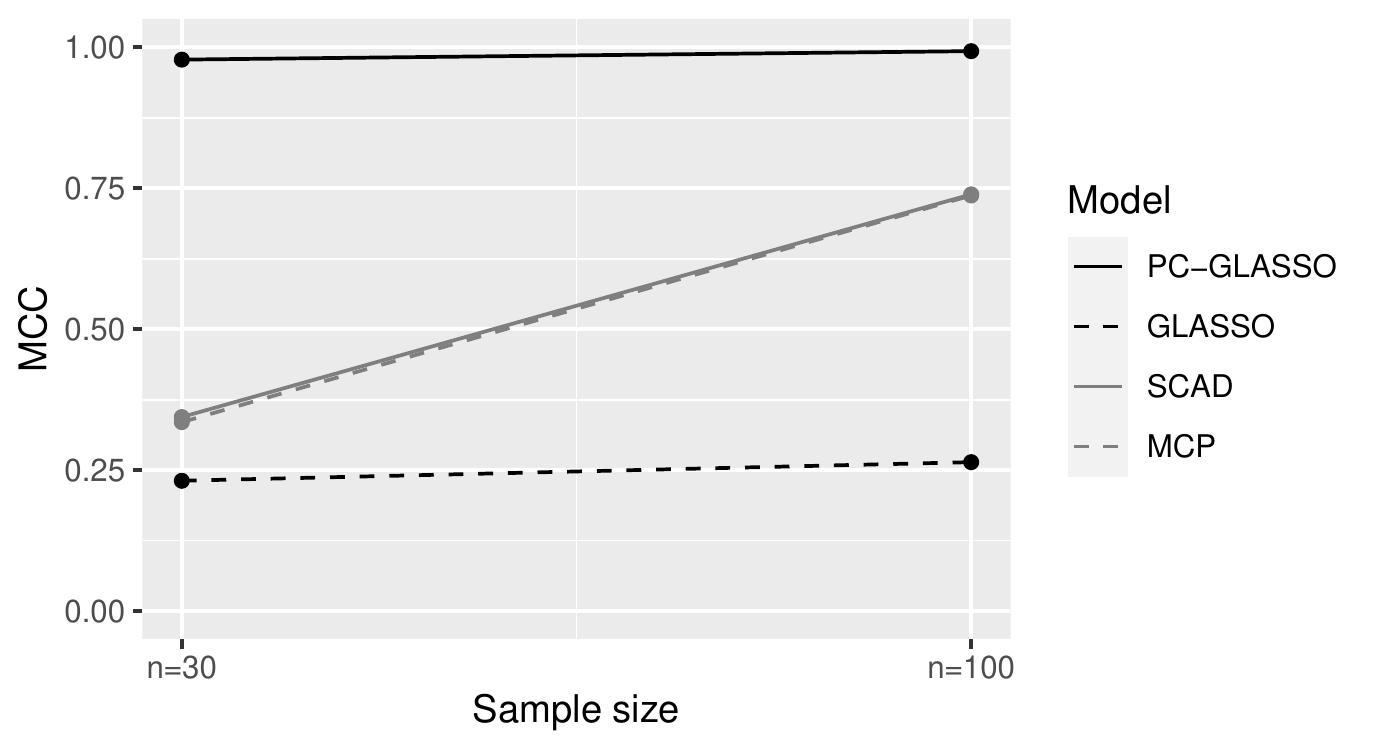} \\
\multicolumn{2}{c}{Hub graph} \\
\includegraphics[scale=0.6]{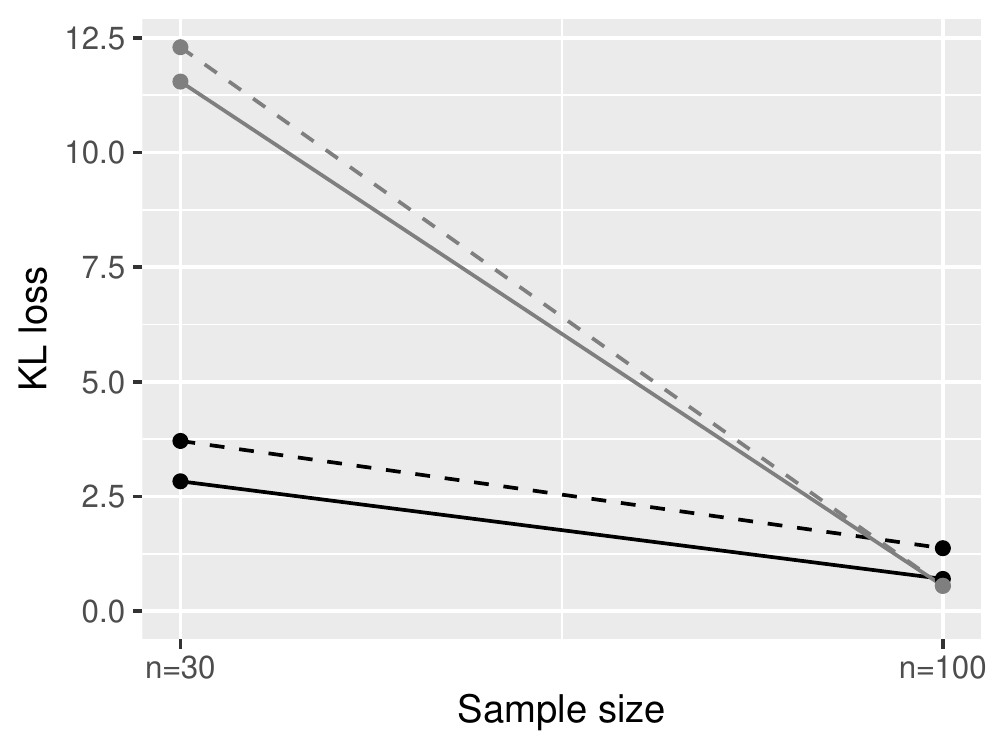} &
	\includegraphics[scale=0.6]{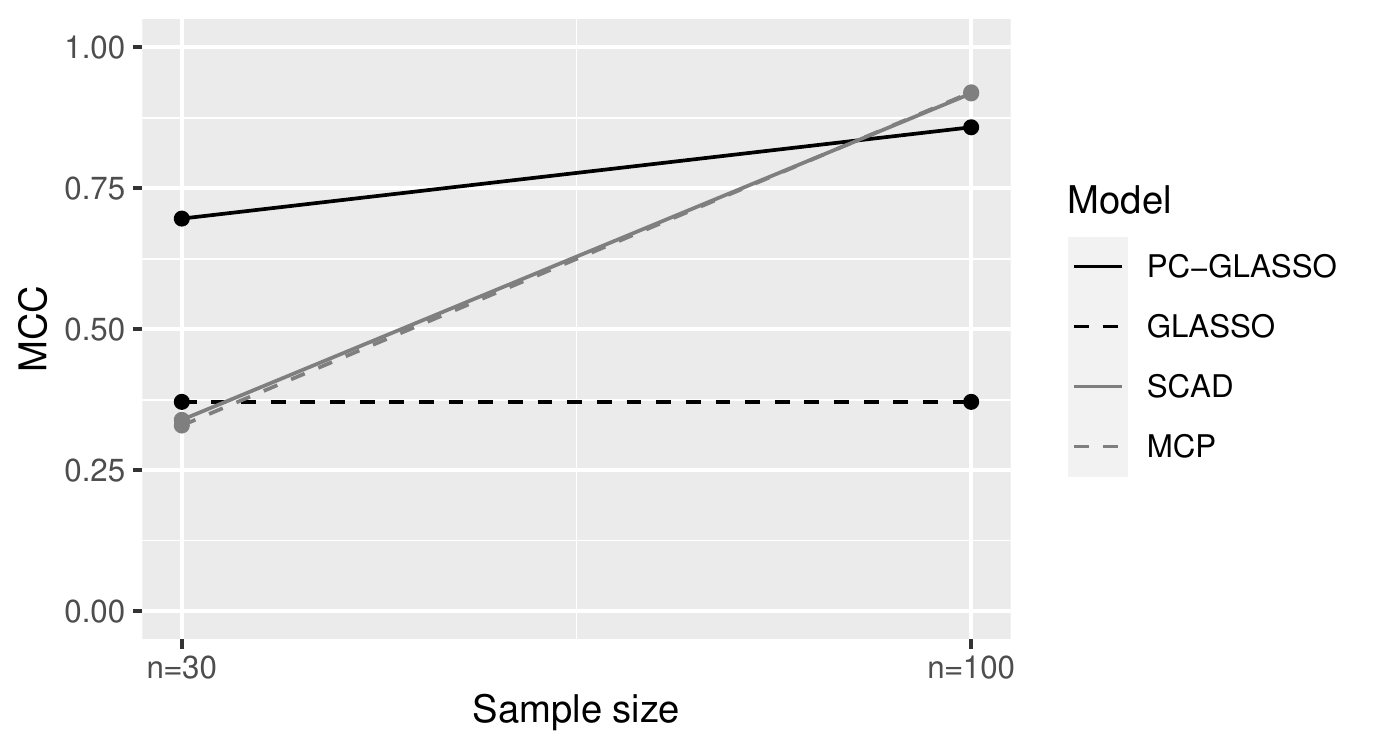} \\
	\multicolumn{2}{c}{AR2 graph} \\
\includegraphics[scale=0.6]{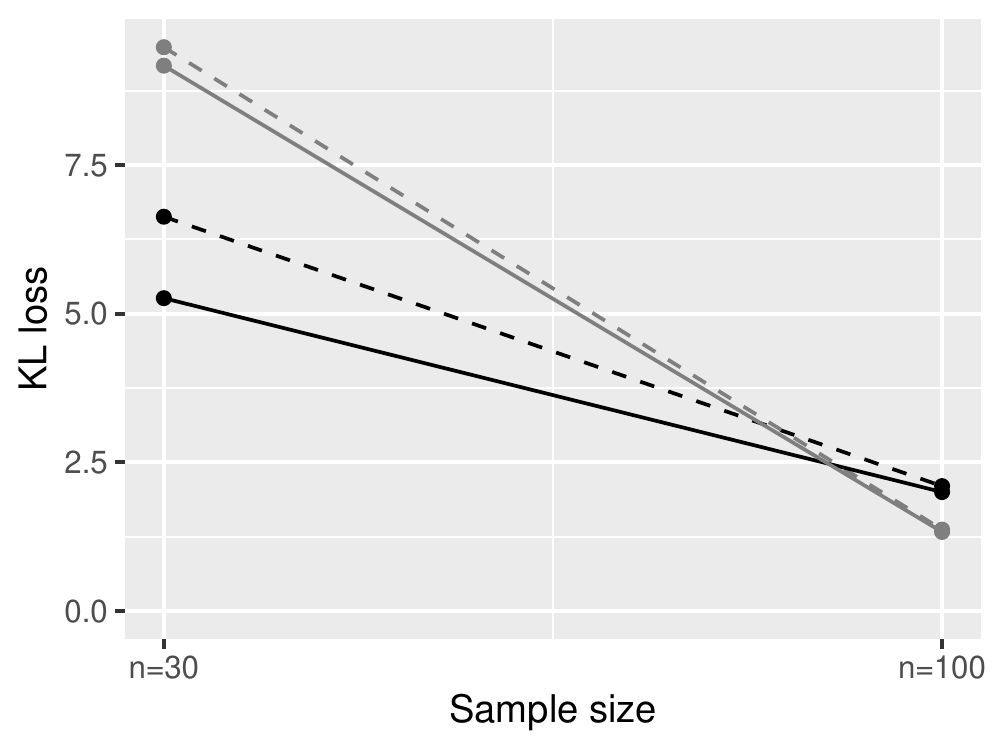} &
	\includegraphics[scale=0.6]{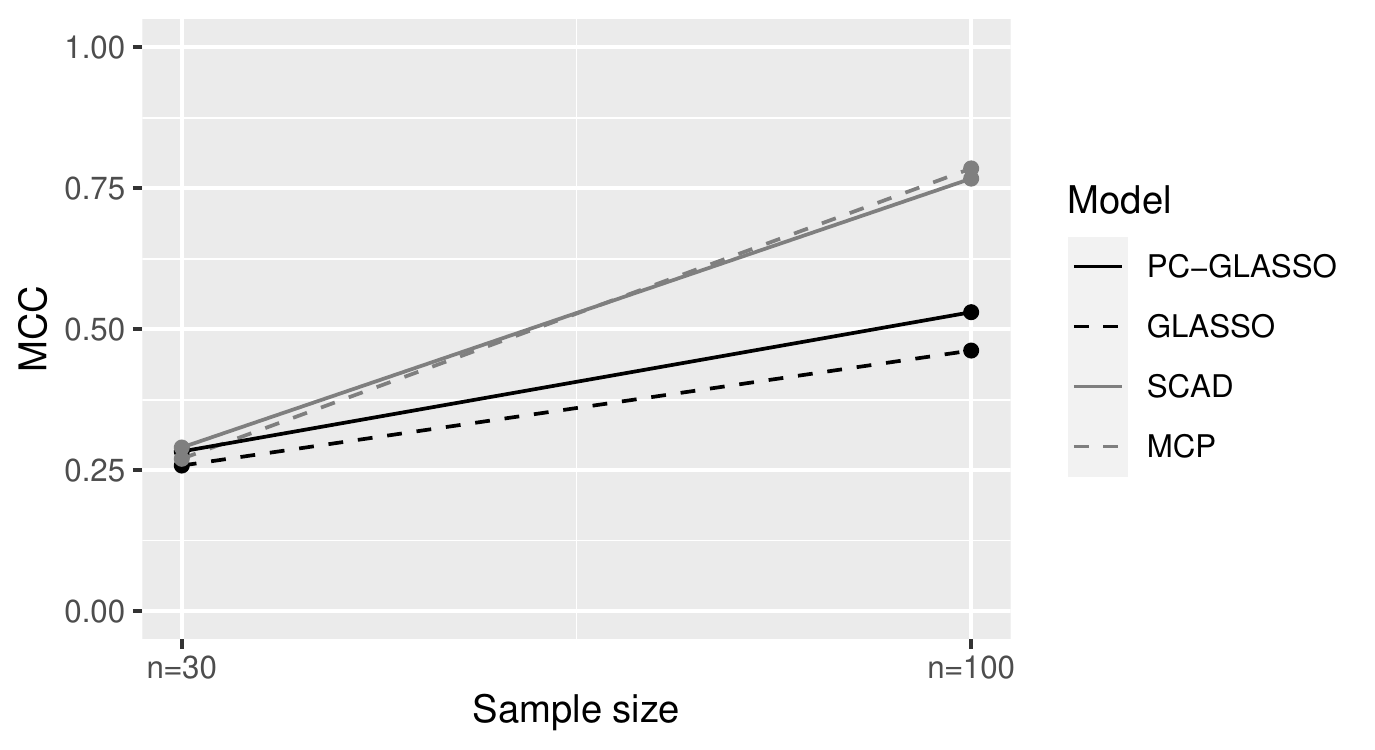} \\
	\multicolumn{2}{c}{Random graph} \\
\includegraphics[scale=0.6]{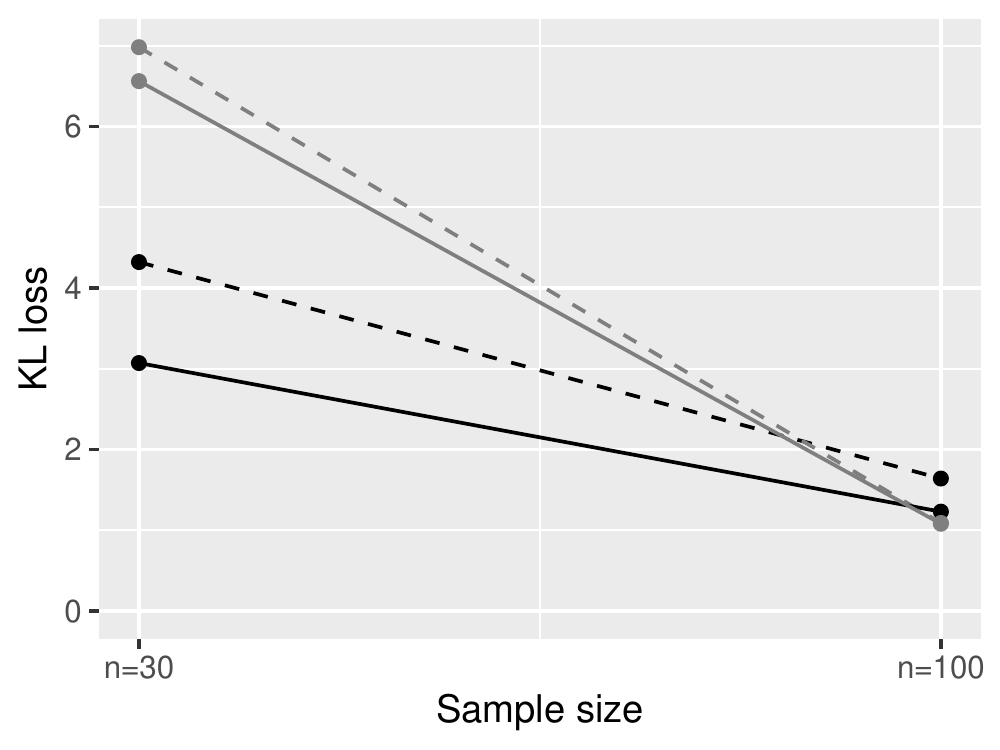} &
	\includegraphics[scale=0.6]{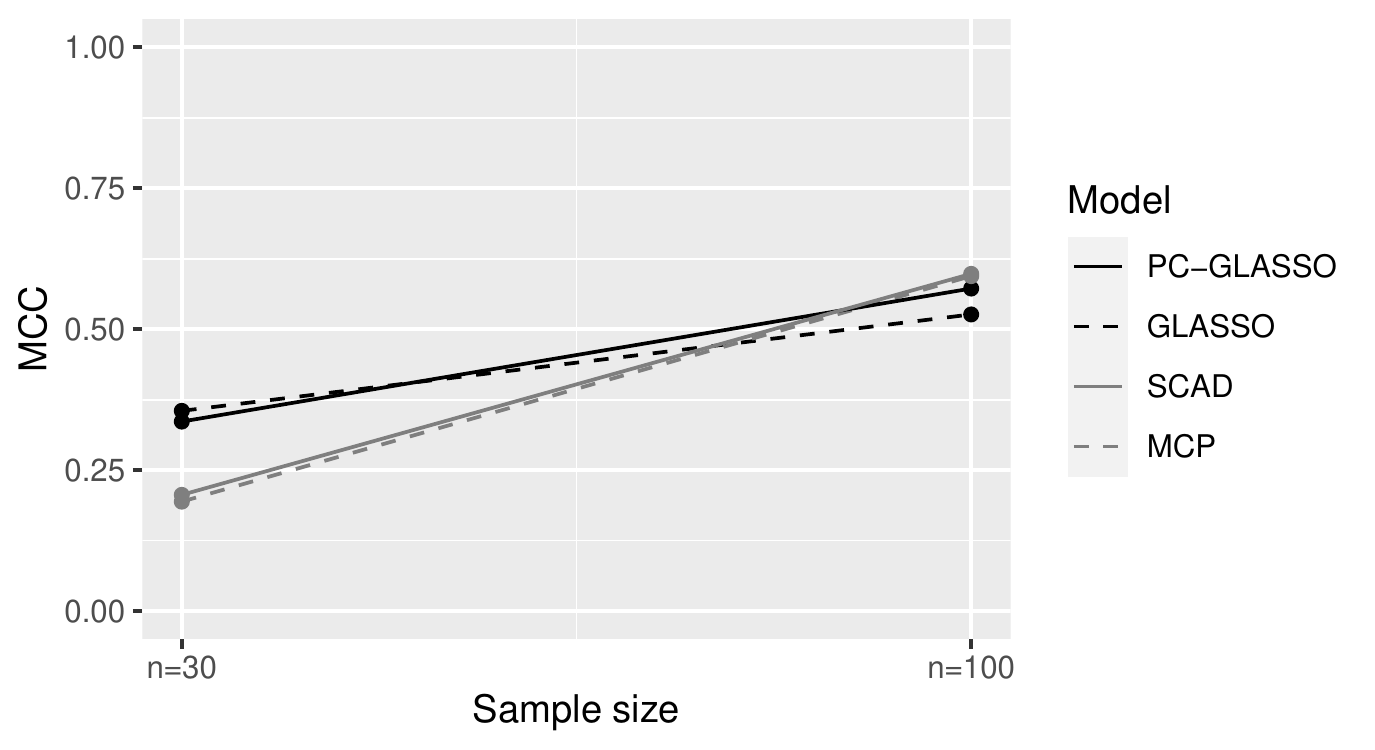}
\end{tabular}
\caption{Kullback-Leibler loss (left) and MCC (right) in the four simulation settings}
\label{fig:allgraphs}
\end{figure}

\begin{figure}[p]
\begin{tabular}{ccccc}
\multicolumn{5}{c}{Star graph} \\
	\includegraphics[scale=0.9]{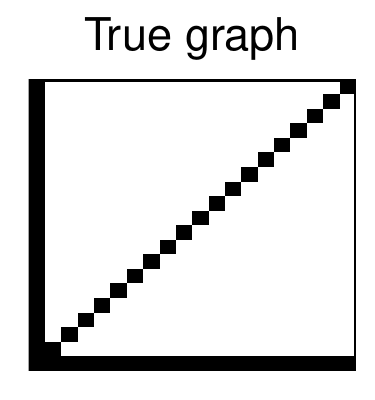} &
	\includegraphics[scale=0.9]{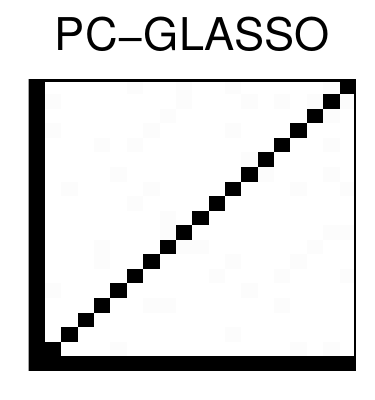} &
	\includegraphics[scale=0.9]{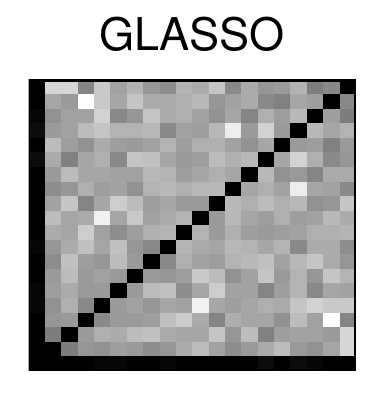} &
	\includegraphics[scale=0.9]{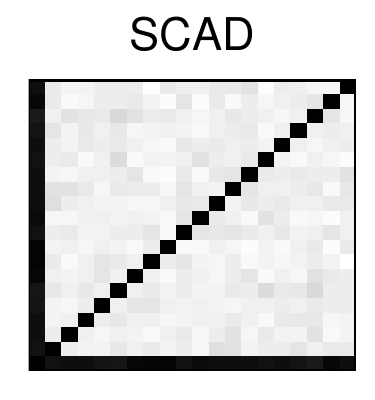} &
	\includegraphics[scale=0.9]{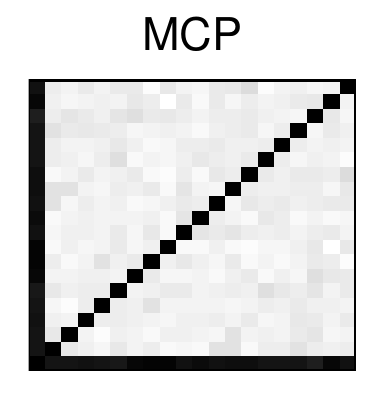} \\
\multicolumn{5}{c}{Hub graph} \\
	\includegraphics[scale=0.9]{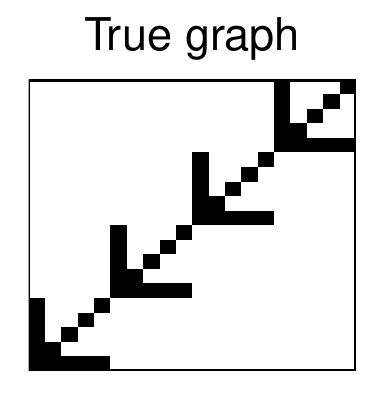} &
	\includegraphics[scale=0.9]{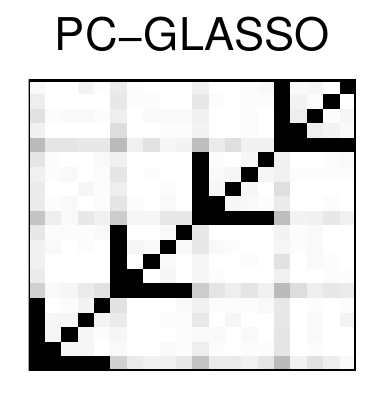} &
	\includegraphics[scale=0.9]{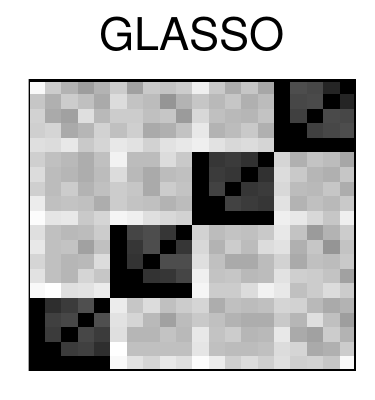} &
	\includegraphics[scale=0.9]{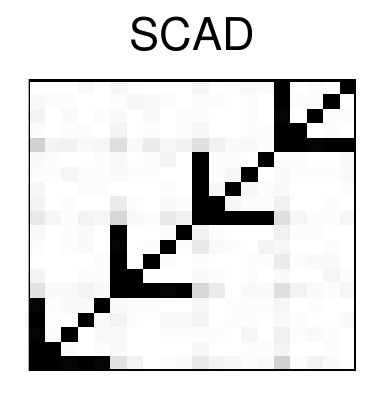} &
	\includegraphics[scale=0.9]{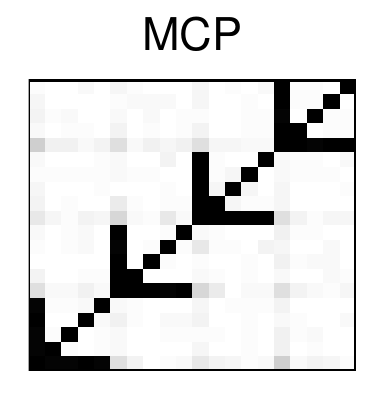} \\
\multicolumn{5}{c}{AR2 graph} \\
	\includegraphics[scale=0.9]{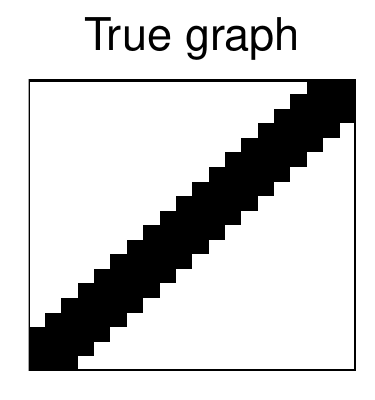} &
	\includegraphics[scale=0.9]{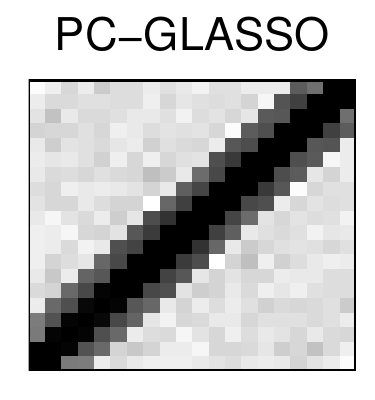} &
	\includegraphics[scale=0.9]{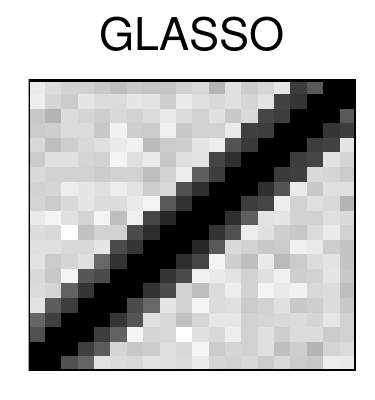} &
	\includegraphics[scale=0.9]{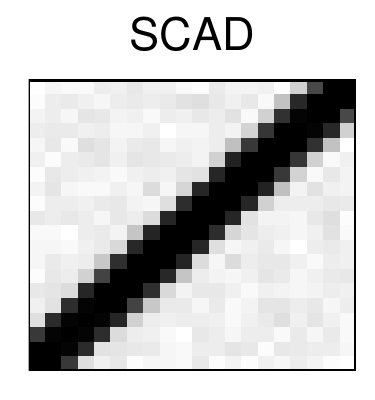} &
	\includegraphics[scale=0.9]{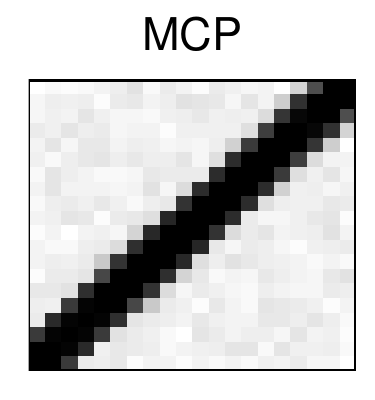} \\
\multicolumn{5}{c}{Random graph} \\
	\includegraphics[scale=0.9]{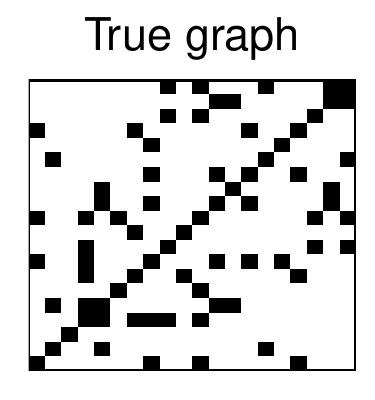} &
	\includegraphics[scale=0.9]{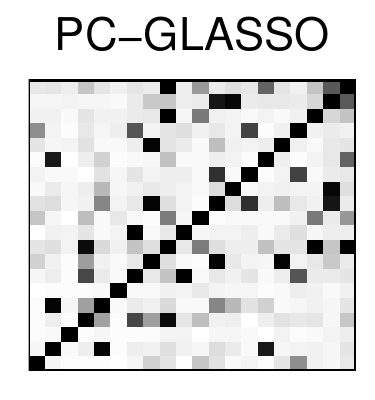} &
	\includegraphics[scale=0.9]{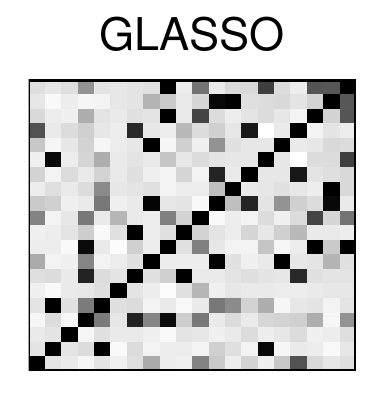} &
	\includegraphics[scale=0.9]{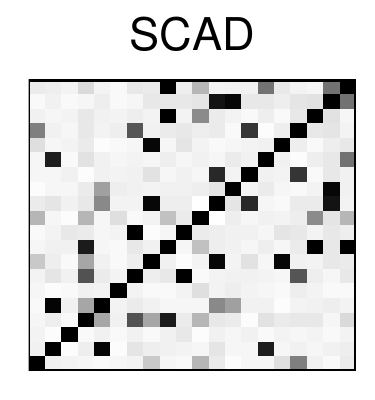} &
	\includegraphics[scale=0.9]{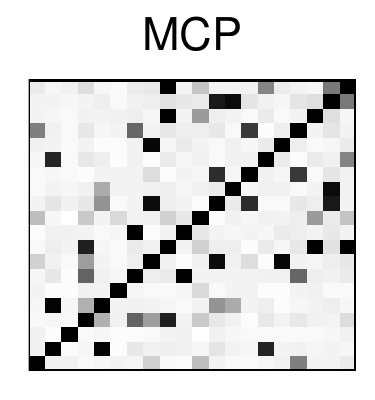}
\end{tabular}
\caption{Proportion of simulations in which each edge was selected}
\label{fig:allgraphs_selection}
\end{figure}

\subsection{Gene expression data}

We assessed the predictive performance of the four penalised likelihood methods in the gene expression data of \cite{Calon2012}. The data contain $262$ observations of $p=173$ genes related to colon cancer progression.  We took $n=200$ of the samples as training data, left the remaining $62$ observations as test data, and
assessed the predictive accuracy of each method by evaluating the log-likelihood on the test data.

Figure \ref{fig:Cancer_StockSizeLike} (left) plots the model size vs. test sample log-likelihood, 
and indicates the models chosen by the BIC and EBIC.  For both these solutions, PC-GLASSO achieved a significantly higher log-likelihood than the other three methods, and selected a model of roughly comparable size.

\begin{figure}[ht]
\vspace{10pt}
\begin{tabular}{cc}
\includegraphics[scale=0.5]{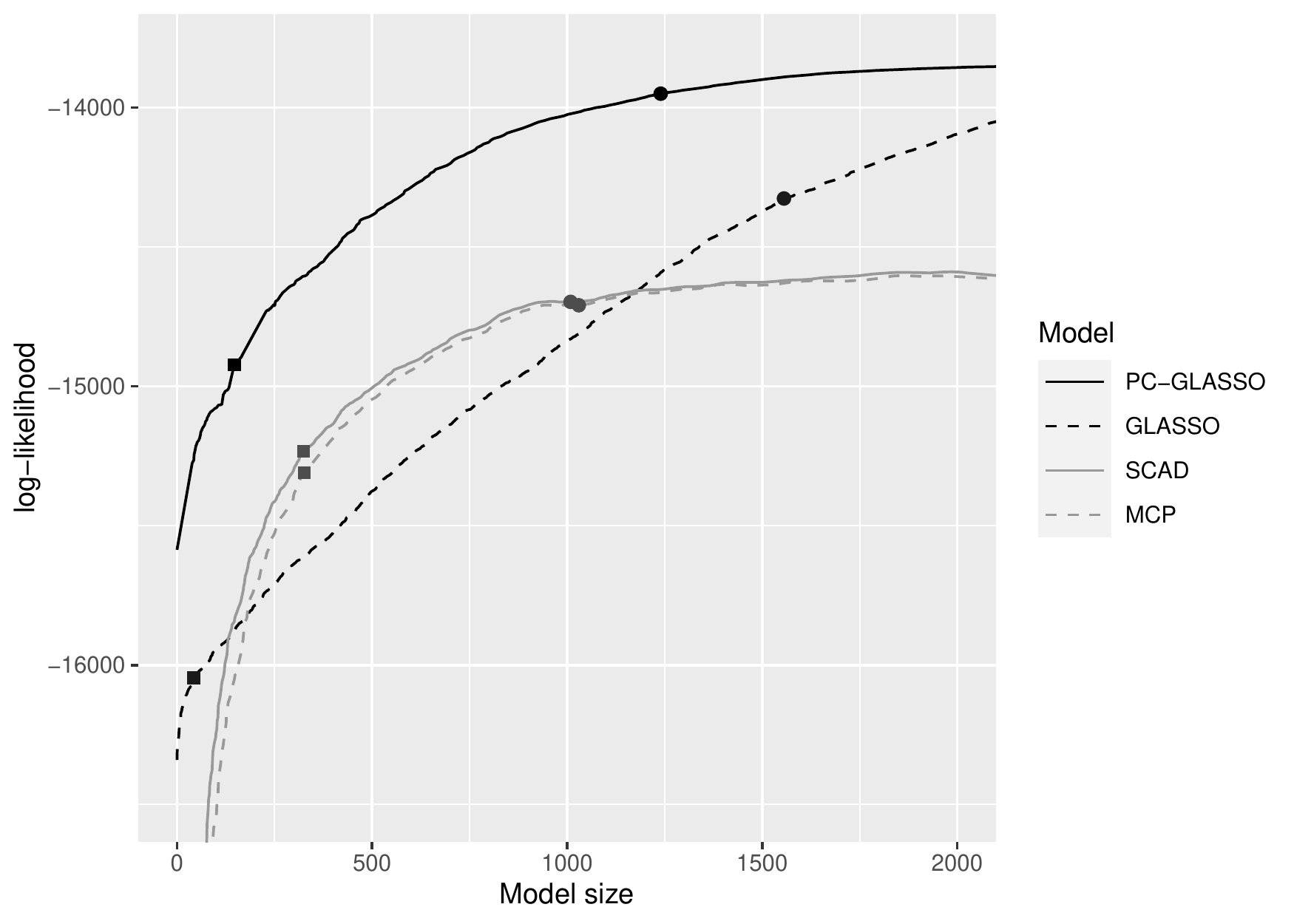} &
	\includegraphics[scale=0.5]{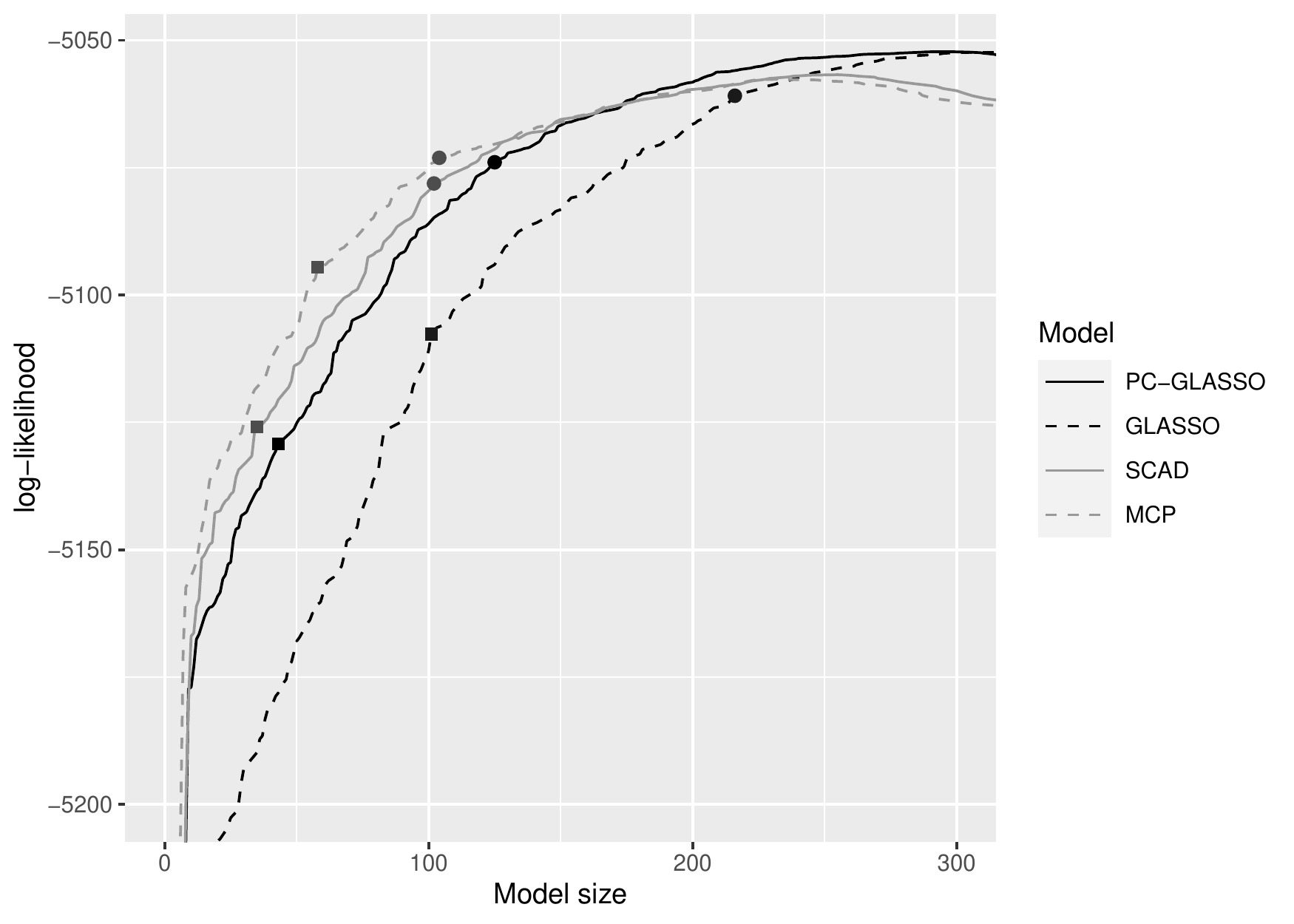}
\end{tabular}
\caption{Model size vs predictive ability in the gene expression (left) and stock market (right) data.  Estimates selected via BIC and EBIC with $\gamma=0.5$ are shown by dots and squares respectively.}
\label{fig:Cancer_StockSizeLike}
\end{figure}

\subsection{Stock market data}

We analyzed the stock market data in the R package \textbf{huge}, investigated in the graphical model context by \cite{Banerjee2015}. 
The data contain daily closing stock prices of companies in the S\&P 500 index between 1st January 2003 and 1st January 2008.
We consider de-trended stock-market log-returns, to study the dependence structure after accounting for the overall mean market behavior. Specifically, let $Y_{jt}$ be the closing price of company $j$ at time $t$, $\tilde{X}_{jt} = \log\left(\frac{Y_{j,t+1}}{Y_{jt}}\right)$ the log-returns, and $X_{jt} = \tilde{X}_{jt} - \bar{X}_{t}$ the de-trended returns, where $\bar{X}_{t} = \sum_{j=1}^p \tilde{X}_{jt}$. 
We randomly selected $p=30$ companies and, to avoid issues with stock market data exhibiting thicker tails than the assumed Gaussian model, we removed outlying observations more than 5 sample standard deviations away from the mean in any of the $p$ variables.
There remained 1,121 observations of which we randomly selected 1,000 for the training and 121 for the test data.  

Figure \ref{fig:Cancer_StockSizeLike} (right) shows the results, which highlight interesting trade-offs in sparsity vs. predictive accuracy. PC-GLASSO selected a smaller model than GLASSO for BIC and EBIC, and achieved a higher log-likelihood in the test data for any model with $<200$ edges, whereas GLASSO attained a higher log-likelihood at the selected model.
Interestingly, the SCAD and MCP penalties provided a similar accuracy to PC-GLASSO, albeit slightly higher for models with $<150$ edges and slightly lower for larger models.

\section{Discussion}\label{sec:discussion}

Penalised likelihood methods based on regular penalty functions are a staple of Gaussian graphical model selection and precision matrix estimation.  They provide a conceptually easy strategy to obtain sparse estimates of $\Theta$ and, particularly in the case of GLASSO, fairly efficient computation,
even for moderately large dimensions.  
However, in this paper we demonstrated that estimates obtained from regular penalties depend on the scale of the variables.  This gives a situation where a simple change of units (measuring a distance in miles rather than kilometers) can result in different graphical model selection. Further, we showed that notions of exchangeability also motivate the need for standardising the data when using regular penalties.

Standardising the data is not innocuous. 
First, even when the variables follow a Gaussian distribution, that is no longer the case for the scaled variables, which exhibit thicker tails.
Second, as demonstrated in several of our examples, applying regular penalties to scaled data can adversely affect inference.
This effect was particularly detrimental in examples where the true underlying graph has a large range in node degrees, as in the Star graph setting.

A wide class of PC-separable penalties, including the PC-GLASSO, overcome these issues as they are scale invariant and do not require standardisation. Using a Bayesian viewpoint, we illustrated 
that PCGLASSO induces a different shrinkage than standard penalties, in that the former induces shrinkage on partial correlations, whereas the latter do not.
Our examples showed that such differential shrinkage can offer significant improvements both in estimation and model selection.

A limitation of our work lies in the computation.  While the efficiency of the coordinate descent algorithm is reasonable in lower dimensions, the computations become impractical for larger $p$.  However, the conditional convexity of the PC-GLASSO problem opens interesting strategies for future improvements.

Further interesting future work is to investigate the theoretical properties of PC-GLASSO, for example model selection consistency, which holds for GLASSO only under certain nontrivial conditions \citep{Ravikumar2009}. The wider set of PC-separable penalties also warrant further exploration, most obviously PC-separable versions of the SCAD and MCP penalties.  On the Bayesian side, a PC-separable version of the spike and slab penalty of \cite{Gan2018} may also be of interest.  Beyond the Gaussian case, penalisation of partial correlations also seems natural for partial correlation graphs in elliptical and transelliptical distributions, see \cite{Rossell2020}.

\section*{Acknowledgements}
JSC is funded by the EPSRC grant EP/L016710/1 as part of the Oxford-Warwick Statistics Programme (OxWaSP).
DR was partially funded by the Europa Excelencia grant EUR2020-112096, Ram\'on y Cajal Fellowship RYC-2015-18544 and
Plan Estatal PGC2018-101643-B-I00.  JQS was supported by the Alan Turing Institute and funded by the Engineering and Physical Sciences Research Council [grant number EP/K03 9628/1].

\newpage

\appendix
\appendixpage

\renewcommand{\themyprop}{S\arabic{myprop}}
\renewcommand{\themycor}{S\arabic{mycor}}

Appendix \ref{sec:Appendix-CoordDesc} outlines the derivation of the coordinate descent algorithm, and presents Algorithms S\ref{alg:RegPath}-S\ref{alg:CoordDesc} to obtain the PC-GLASSO solution for a sequence of penalisation parameters and a given penalisation parameter value, respectively.
Appendix \ref{sec:Appendix-Proofs} provides the proofs for all our propositions and further results.
Appendix \ref{sec:Appendix-Results} shows some supplementary results for the examples in Section \ref{sec:Simulations}.

\section{Coordinate descent algorithm}\label{sec:Appendix-CoordDesc}

We present the coordinate descent algorithm we used to calculate PC-GLASSO estimates in the simulated examples of this paper.  Our aim is to find the values of $\Theta$ that maximise the objective function (\ref{eq:PenLikeFn}) for a sequence of penalty parameters $0=\rho_{0}<\rho_{1}<\dots<\rho_{k}$, i.e. the regularisation path.  Algorithm S\ref{alg:RegPath}, for which the coordinate descent algorithm S\ref{alg:CoordDesc} is embedded, ensures that the previous estimate related to $\rho_{i-1}$ is used as a starting point for the coordinate descent for $\rho_{i}$.  This ensures that the coordinate descent is initialised at a point close to the maximum and aids convergence.  We also standardise the sample covariance $S$ to have unit diagonals, before returning the estimates to the original scale.  This has no effect on the estimated values due to the scale invariance of PC-GLASSO, however it helps with the numerics of the coordinate descent.

Algorithm S\ref{alg:CoordDesc} is a standard blockwise coordinate descent algorithm which randomly cycles through the entries of $\Delta$ and maximises the objective function with respect to $\Delta_{ij},\Delta_{ji},\theta_{ii},\theta_{jj}$ while holding all other entries fixed.  Once the algorithm has cycled through each of the entries of $\Delta$ exactly once, a stopping rule is tested.  The stopping rule we choose is based on the increase in the value of the objective function brought about by the updates.  If the increase in the objective function is less than a particular threshold then the algorithm is terminated and the current estimate is returned.  Note that the threshold here is scaled by $ q = \max\left\{ \frac{ 2 | \{ \Delta_{ij}^{(0)} \neq 0 : i < j \} | }{ p (p-1) }, \frac{2}{p(p-1)} \right\}$, the proportion of non-zero entries in the previous estimate $\Delta^{(0)}$.  This is because once an entry is shrunk to zero, it is likely that it will remain zero in future estimates.  Therefore, the number of entries that are actively being updated is proportional to $q$.  If only a small number of entries are being actively updated then one would expect the increase in the objective function to be smaller.  Hence, scaling the threshold by $q$ helps to prevent the algorithm from terminating too early in situations where the current estimate is sparse.

Although no guarantees are made about the convergence of Algorithm S\ref{alg:CoordDesc}, results in \cite{Patrascu2015} and \cite{Wright2015} suggest that convergence towards a local maximum is guaranteed and give reasonable assurance of convergence towards the global maximum. Their results focus on a coordinate descent algorithm that cycles randomly through the indices \textit{with} replacement and so are not directly applicable to Algorithm S\ref{alg:CoordDesc}.  However, we prefer cycling through the indices without replacement since this provides a more simple and clear stopping rule for the algorithm.  Algorithm S\ref{alg:CoordDesc} assesses the convergence after updating each entry of $\Delta$ exactly once, so that the stopping rule at the end of each iteration is made on the same grounds. For an algorithm which selects indices with replacement it is less clear when to enact the stopping rule.

As a final note about Algorithm S\ref{alg:CoordDesc}, Step 2 maximising (\ref{eq:PenLikeFn}) with respect to $\Delta_{ij},\theta_{ii},\theta_{jj}$ whilst all other variables are held fixed is non-trivial due to the non-smoothness of the objective function.  The remainder of this section will focus on solving this maximisation problem.  To ease notation let $x=\Delta_{ij}$, $y_{1}=\sqrt{\theta_{ii}}$ and $y_{2}=\sqrt{\theta_{jj}}$.  The objective function is
\begin{dmath*}
f(x,y_{1},y_{2}) = \log( ax^2 + bx + c ) + 2c_{n}( \log(y_{1}) + \log(y_{2}) ) - y_{1}^{2} - y_{2}^{2} - 2c_{12}xy_{1}y_{2} - 2c_{1}y_{1} - 2c_{2}y_{2} - 2\rho|x|,
\end{dmath*}
where
$$ c_{n} = 1 - \frac{4}{n}, $$
$$ c_{12} = S_{ij}, $$
$$ c_{1} = \sum_{k \neq i,j} S_{ik} \Delta_{ik} \sqrt{ \theta_{kk} }, $$
$$ c_{2} = \sum_{k \neq i,j} S_{jk} \Delta_{jk} \sqrt{ \theta_{kk} }. $$
The $\log( ax^2 + bx + c )$ term comes from the $\log \det(\Delta)$, since the determinant of a symmetric matrix is quadratic in the off-diagonal entries.  
The coefficients $(a,b,c)$ do not have a simple closed-form, as they depend on the matrix determinant, but they can be easily obtained by evaluating the determinant of $\Delta$ for three different values of $\Delta_{ij}$ (faster methods for computing these determinants are possible since they only involve changing a single entry) and solving the resulting system of equations.  The range of values that $x$ is can take given by
$$ (l,u) := \{ x : ax^2 + bx + c > 0 \} \cap (-1,1). $$
Any value of $x$ in this set ensures positive definiteness of $\Delta$.  This is because $\Delta$ is positive definite if and only if all its leading principal minors are positive.  WLOG, letting $\Delta_{ij}$ be in the bottom row of $\Delta$, if the previous estimate is positive definite then the first $p-1$ leading principal minors are positive. The condition $ax^2 + bx + c > 0$ ensures that the final leading principal minor, $\det(\Delta)$, is also positive.
The maximisation problem can then be expressed as
\begin{equation}
\begin{aligned}
\max_{x,y_{1},y_{2}} \quad & f(x,y_{1},y_{2})\\
\textrm{s.t.} \quad & x \in (l,u) \\
  & y_{1}, y_{2} > 0    \\
\end{aligned}
\end{equation}

We denote the partial derivatives of $f$ by
$$ f_{x}(x,y_{1},y_{2}) = \frac{ 2ax + b }{ ax^2 + bx + c } - 2c_{12}y_{1}y_{2} - 2\rho\mathrm{sign}(x), \quad x\neq 0, $$
$$ f_{y_{1}}(x,y_{1},y_{2}) = 2c_{n}y_{1}^{-1} - 2y_{1} - 2c_{12}xy_{2} - 2c_{1}, $$
$$ f_{y_{2}}(x,y_{1},y_{2}) = 2c_{n}y_{2}^{-1} - 2y_{2} - 2c_{12}xy_{1} - 2c_{2}, $$

To solve this problem we consider separately the cases $c>0$ and $c \leq 0$.

\subsection*{Case $c>0$.}

We begin by looking at the case $c>0$, which implies that $0 \in (l,u)$.  We split the problem into three sections, finding local maxima in $x=0$, $x\in(0,u)$, $x\in(l,0)$ separately and then selecting from these the global maximum.

\subsubsection*{Optimization for $x=0$.}
Let $x=0$.  By setting $f_{y_{1}}(x,y_{1},y_{2}) = 0$ and $f_{y_{2}}(x,y_{1},y_{2}) = 0$ we get that the optimal values of $(y_1,y_2)$ are
$$ y_{1} = \frac{1}{2} \left( \sqrt{ c_{1}^{2} + 4c_{n} } - c_{1} \right), $$
$$ y_{2} = \frac{1}{2} \left( \sqrt{ c_{2}^{2} + 4c_{n} } - c_{2} \right). $$

\subsubsection*{Optimization over $x>0$.}
Let $x \in (0,u)$.  Setting $f_{y_{1}}(x,y_{1},y_{2})=0$ gives
\begin{equation}\label{eq:x}
x = \frac{ c_{n} y_{1}^{-1} - y_{1} - c_{1} }{ c_{12} y_{2} },
\end{equation}
and setting $f_{y_{2}}(x,y_{1},y_{2})=0$ along with (\ref{eq:x}) gives
\begin{equation}\label{eq:y2}
y_{2} = \frac{1}{2}\left( - c_{2} \pm \sqrt{ c_{2}^{2} + 4( y_{1}^{2} + c_{1} y_{1} } ) \right).
\end{equation}
Using (\ref{eq:x})-(\ref{eq:y2}) one can write $f_{x}(x,y_{1},y_{2})$ in terms of only $y_{1}$ and solve $f_{x}(x,y_{1},y_{2})=0$ numerically to obtain the stationary points.  The range of $y_{1}$ values to search in the numerical solving of $f_{x}(x,y_{1},y_{2})=0$ can be found by considering the constraints $x \in (0,u)$, $y_{1},y_{2}>0$ as well as (\ref{eq:x}) and (\ref{eq:y2}).

The constraint $x<u$ results in some condition on the following quartic which we refer to as $q(y_1)$
\begin{align}\label{eq:quartic}
\left( 1 - \frac{ 1 }{ u^2 c_{12}^2 } \right) y_1^4 + 
\left( c_1 - \frac{ 2c_1 }{ u^2 c_{12}^2 } + \frac{ c_2 }{ u c_{12} } \right) y_1^3 +
\left( \frac{ 2c_n }{ u^2 c_{12}^2 } - \frac{ c_1^2 }{ u^2 c_{12}^2 } + \frac{ c_1 c_2 }{ u c_{12} } \right) y_1^2 +
\left( \frac{ 2 c_1 c_n }{ u^2 c_{12}^2 } - \frac{ c_2 c_n }{ u c_{12} } \right) y_1 -
\frac{ c_n^2 }{ u^2 c_{12}^2 }
\end{align}

We first summarize the range of $y_1$ values that needs to be considered, depending on the values of $(c_{12},c_2)$, and subsequently outline their derivation.  If the positive root is taken in (\ref{eq:y2}) for $y_2$ then the following constraints are required
\begin{enumerate}
\item $y_{1} < \frac{1}{2}\left( - c_{1} + \sqrt{ c_{1}^{2} + 4 c_{n} } \right) \mbox{, if } c_{12}>0$

\item $y_{1} > \frac{1}{2}\left( - c_{1} + \sqrt{ c_{1}^{2} + 4 c_{n} } \right) \mbox{, if } c_{12}<0$

\item $y_{1} \geq \frac{1}{2}\left(-c_1 + \sqrt{ c_{1}^{2} - c_{2}^{2} } \right)
\mbox{ or }
y_{1} \leq \frac{1}{2}\left(-c_1 - \sqrt{ c_{1}^{2} - c_{2}^{2} } \right)$

\item $y_1 > -c_1 \mbox{, if } c_2 > 0$

\item If $c_{12}>0$, either $y_1 > \frac{1}{2} \left( \frac{1}{2} u c_{12} c_2 - c_1 + \sqrt{ \left( c_1 - \frac{1}{2} u c_{12} c_2 \right)^2 + 4c_n } \right)$ or $q(y_1)>0$

\item If $c_{12}<0$, either $y_1 < \frac{1}{2} \left( \frac{1}{2} u c_{12} c_2 - c_1 + \sqrt{ \left( c_1 - \frac{1}{2} u c_{12} c_2 \right)^2 + 4c_n } \right)$ or $q(y_1)<0$
\label{eq:range_y1_cda}
\end{enumerate}

The negative root in (\ref{eq:y2}) must only be considered if $c_2<0$ and $y_1<-c_1$ (also implying that $c_1<0$ and, from constraint 1, $c_{12}>0$).  In this case the inequalities in constraints 5 and 6 must be reversed.

We outline how to obtain the above constraints. The constraint $x > 0$ along with (\ref{eq:x}) implies that
$$ \mathrm{sign}( y_{1}^{2} + c_{1}y_{1} - c_{n} ) = - \mathrm{sign}( c_{12} ). $$
Hence, if $c_{12} > 0$ then the range of values to consider can be restricted to
$$ y_{1} < \frac{1}{2}\left( - c_{1} + \sqrt{ c_{1}^{2} + 4 c_{n} } \right), $$
giving constraint 1, while if $c_{12} < 0$ then the inequality is reversed giving constraint 2.  Note that if $c_{12}=0$ then the optimisation problem is simpler and so the details of this case are omitted.

For $y_{2}$ to take a real value in (\ref{eq:y2}) we must have $4y_{1}^{2} + 4c_{1}y_{1} + c_{2}^{2} \geq 0$ which implies that either
$$ y_{1} \geq \frac{1}{2}\left( \sqrt{ c_{1}^{2} - c_{2}^{2} } - c_{1} \right),$$
or
$$ y_{1} \leq \frac{1}{2}\left( - \sqrt{ c_{1}^{2} - c_{2}^{2} } - c_{1} \right).$$
giving constraint 3.

Combining the constraint $y_{2}>0$ with (\ref{eq:y2}), if $c_{2} > 0$ then we need $y_{1} \geq -c_{1}$ in order for there to be a solution for $y_{2}$, giving constraint 4.  On the other hand, if $c_{2}<0$ and $0<y_{1}<-c_{1}$ then there are two solutions for $y_{2}$ and one must consider both the positive and negative roots in (\ref{eq:y2}).  For all other situations one must only consider the positive root.

Now combining the constraint $x<u$ with (\ref{eq:x}) and (\ref{eq:y2}), one obtains the inequality
$$ \frac{2}{ u c_{12} } \left( c_n y_1^{-1} - y_1 - c_1 \right) + c_2 < \sqrt{ c_2^2 + 4( y_1^2 + c_1y_1 ) } $$
from which constraints 5 and 6 follow.

Combining each of these constraints give the range of possible values for $y_{1}$ to numerically search for a stationary point.  Once $y_{1}$ is found, (\ref{eq:y2}) and (\ref{eq:x}) give the corresponding $(x,y_2)$. Note that it is possible that there be no stationary points within $x>0$.

\subsubsection*{Optimization over $x<0$.}

Finding stationary points in the interval $x \in (l,0)$ is analogous to the case where $x \in (0,u)$, but with some sign changes and so the details are omitted.

\subsection*{Case $c \leq 0$.}

Consider the case where $c \leq 0$.  Then it is easy to see that when $b > 0$ then $(l,u) \subseteq (0,1)$, while if $b<0$ then $(l,u) \subseteq (-1,0)$.  Again, solving this is very similar to the previous case, however one must pay closer attention to the range of values $y_{1}$ may take.  In particular, when $b > 0$, (\ref{eq:x}) must still hold at stationary points, but one must restrict this in $(l,u)$ rather than $(0,u)$.  This results in two quartic constraints on $y_{1}$.  Again the details are omitted.

\begin{algorithm}[H]\label{alg:RegPath}
\SetKwInOut{Input}{Input}\SetKwInOut{Output}{Output}
\SetAlgoLined
\Input{ Sample covariance $S$, sequence of penalty parameters $0=\rho_{0}<\rho_{1}<\dots,\rho_{k}$ and optimisation convergence threshold $\epsilon$. }
\Output{ Sequence of estimates $\hat{\Theta}_0,\ldots,\hat{\Theta}_k$ corresponding to $\rho_0,\ldots,\rho_k$. }
\begin{enumerate}
\item Standardise the sample covariance $\tilde{S} =  \mathrm{diag}(S)^{-1/2} S \mathrm{diag}(S)^{-1/2}$.
\item Run Algorithm S\ref{alg:CoordDesc} on $\tilde{S}$ for $\rho=0$, with starting point $\Theta_0^{(0)}=\tilde{S}^{-1}$ (or Moore-Penrose inverse if $n<p$), and threshold $\epsilon$ to obtain an estimate $\tilde{\Theta}_0$.
\item For $i=1,\dots,k$, run Algorithm S\ref{alg:CoordDesc} on $\tilde{S}$ for penalty parameter $\rho=\rho_{i}$, with starting point $\Theta_i^{(0)}=\tilde{\Theta}_{(i-1)}$, and threshold $\epsilon$ to obtain an estimate $\tilde{\Theta}_i$.
\item Return the sequence of estimates $\Theta_i = \mathrm{diag}(S)^{-1/2} \tilde{\Theta}_i \mathrm{diag}(S)^{-1/2}$ for $i=0,1,\dots,k$.
\end{enumerate}
 \caption{PC-GLASSO regularisation path}
\end{algorithm}

\begin{algorithm}[H]\label{alg:CoordDesc}
\SetKwInOut{Input}{Input}\SetKwInOut{Output}{Output}
\SetAlgoLined
\Input{ Sample covariance $S$ with unit diagonal, penalty parameter $\rho$, start point $\Theta^{(0)}$ and optimisation convergence threshold $\epsilon$. }
\Output{ A matrix $\Theta$ providing a local maximum of (\ref{eq:PenLikeFn}) for penalty $\rho$. }
\begin{enumerate}
\item Let $\Theta^{(1)}=\Theta^{(0)}$ and decompose $\Theta^{(1)}$ to get $\theta^{(1)}$ and $\Delta^{(1)}$.
\item Cycling randomly without replacement through the set of indices $\{ (i,j) : i < j; i,j\in \{ 1,\dots,p \} \}$, let $\Delta_{ij},\theta_{ii},\theta_{jj}$ maximise
$$f( \Delta, \theta ) = \log( \det( \Delta ) ) + \left( 1 - \frac{ 4 }{ n } \right) \sum_{i} \log( \theta_{ii} ) - \mathrm{tr}\left( S \theta^{1/2} \Delta \theta^{1/2} \right) - \rho \Vert \Delta \Vert_{1},$$
subject to
$$ \Delta \in \mathcal{S}_{1}, $$
$$ \Delta_{k_{1}k_{2}} = \Delta_{k_{1}k_{2}}^{(1)},  \text{  for all } (k_{1},k_{2}) \neq (i,j), $$
$$ \theta_{ii},\theta_{jj} \geq 0, $$
$$ \theta_{kk} = \theta_{kk}^{(1)},  \text{  for all } k \neq i,j, $$
and update $\Delta_{ij}^{(1)}=\Delta_{ij}$, $\Delta_{ji}^{(1)}=\Delta_{ji}$, $\theta_{ii}^{(1)}=\theta_{ii}$, $\theta_{jj}^{(1)}=\theta_{jj}$.
\item Let $ q = \max\left\{ \frac{ 2 | \{ \Delta_{ij}^{(0)} \neq 0 : i < j \} | }{ p (p-1) }, \frac{2}{p(p-1)} \right\}$ be the proportion of non-zero off-diagonal entries.
\item If $f(\Delta^{(1)},\theta^{(1)}) - f(\Delta^{(0)},\theta^{(0)}) < q \epsilon$, set $\Delta=\Delta^{(1)}$, $\theta=\theta^{(1)}$ and return $\Theta=\theta^{1/2}\Delta\theta^{1/2}$.  Otherwise, set $\Delta^{(0)}=\Delta^{(1)}$, $\theta^{(0)}=\theta^{(1)}$ and return to Step 2.
\end{enumerate}
 \caption{Blockwise coordinate descent}
\end{algorithm}

\section{Proofs}\label{sec:Appendix-Proofs}

In this section we present the proofs for each of the results in this paper as well as some supplementary results.

\subsection{Mean squared error of logarithmic penalty}\label{subsec:Appendix-MSE}

This section addresses the claim of Section \ref{sec:PCGLASSO} related to the mean squared error of logarithmic penalties in the $p=1$ case.
Specifically, we show that amongst penalty functions of the form $c\log(x)$ for constant $c \geq 0$ on the precision, choosing $c=2$ asymptotically minimises the mean squared error of the estimate of the precision.

Suppose we have $n$ observations of $X \sim \mathrm{N}( \mu, \theta^{-1} )$ with sample variance $s$.  Note that
$$ (n - 1) \theta s \sim \chi^{2}_{n-1}, $$
and so
$$ \left( (n - 1) \theta s \right)^{-1} \sim \mathrm{Inv}-\chi^{2}_{n-1}. $$
From this we get that
$$ \mathbb{E}[ s^{-1} ] = \frac{ n - 1 }{ n - 3 } \theta,$$
$$ \mathrm{Var}( s^{-1} ) = \frac{ 2 ( n - 1 )^2 }{ ( n - 3 )^2 ( n - 5 ) } \theta^{2}. $$
Consider estimating $\theta$ via a penalised likelihood of the form
$$ l( \theta \mid s ) - c \log( \theta ).$$
This can easily be shown to be maximised at 
$$\hat{\theta} = \left( 1 - \frac{2c}{n} \right) s^{-1}.$$
It follows that
$$ \mathbb{E}[ \hat{\theta} ] = \left( 1 - \frac{2c}{n} \right) \left( \frac{ n - 1 }{ n - 3 } \right) \theta, $$
$$ \mathrm{Var}( \hat{\theta} ) = \frac{ 2 ( 1 - \frac{2c}{n} )^2 ( n - 1 )^2 }{ ( n - 3 )^2 ( n - 5 ) } \theta^{2}, $$
and so
\begin{align*}
\mathrm{MSE}( \hat{\theta} ) &= \mathrm{Var}( \hat{\theta} ) + \left( \mathbb{E}[ \hat{\theta} ] - \theta \right)^2 \\
&= \theta^2 \left( \frac{ 2 ( 1 - \frac{2c}{n} )^2 ( n - 1 )^2 }{ ( n - 3 )^2 ( n - 5 ) } + \left( \left( 1 - \frac{2c}{n} \right) \left( \frac{ n - 1 }{ n - 3 } \right) - 1 \right) ^2 \right)
\end{align*}
It can be shown that this function is minimised at $c=\frac{2n}{n-1}$.  Letting $n \rightarrow \infty$ we get that the MSE is asymptotically minimised amongst logarithmic penalties by taking $c=2$.

\subsection{Proofs for Section \ref{sec:ScaleInv}}\label{subsec:Appendix-ScaleInvProofs}

\begin{proof}[Proof of Proposition \ref{prop:RegularScaleInv}]

Let $S$ be some sample covariance matrix for which $\hat{\Theta}(S)$ is not diagonal and $D$ be some diagonal matrix with non-zero diagonal entries $d_{i}$, $i=1,\dots,p$.  Suppose that $\hat{\Theta}$ is scale invariant.  Let $\hat{\theta}_{ij} = \hat{\Theta}(S)_{ij}$ be some non-zero off-diagonal entry of $\hat{\Theta}(S)$, and $\tilde{\theta}_{ij} = \hat{\Theta}(DSD)_{ij}$ be the corresponding entry in $\hat{\Theta}(DSD)$.  By scale invariance we must have $\tilde{\theta}_{ij} = \frac{\hat{\theta}_{ij}}{d_{i}d_{j}}$.

For these to maximise their corresponding penalised likelihoods, the derivatives of the penalised likelihood function (\ref{eq:PenLike}) with respect to $\theta_{ij}$ must be equal to $0$ at $\hat{\theta}_{ij}$ and $\tilde{\theta}_{ij}$ respectively (note that the derivative exists because $Pen$ is regular and $\hat{\theta}_{ij} \neq 0, \tilde{\theta}_{ij} \neq 0$). 
Therefore
$$ (\hat{\Theta}(S)^{-1})_{ij} - 2s_{ij} - \frac{4}{n} pen_{ij}'(\hat{\theta}_{ij}) = 0, $$
\begin{align*}
(\hat{\Theta}(DSD)^{-1})_{ij} - 2d_{i}d_{j}s_{ij} - \frac{4}{n} pen_{ij}'(\tilde{\theta}_{ij}) &= d_{i}d_{j}\left((\hat{\Theta}(S)^{-1})_{ij} - 2s_{ij}\right) - \frac{4}{n} pen_{ij}'\left(\frac{\hat{\theta}_{ij}}{d_{i}d_{j}}\right) \\
&= 0,
\end{align*}
where we used that, since $\hat{\Theta}$ is scale invariant then $\hat{\Theta}(D S D)= D^{-1} \hat{\Theta}(S) D^{-1}$ and 
hence $(\hat{\Theta}(DSD)^{-1})_{ij}= (D \hat{\Theta}(S)^{-1} D)_{ij}= d_{i}d_{j} (\hat{\Theta}(S)^{-1})_{ij}$.

It follows that
\begin{align} 
pen_{ij}'\left(\frac{\hat{\theta}_{ij}}{d_{i}d_{j}}\right) = d_{i}d_{j}pen_{ij}'(\hat{\theta}_{ij}).
\label{eq:necessarycond_scaleinv}
\end{align}
That is, for scale invariance to hold
the penalty must satisfy $pen_{ij}'\left(\frac{\hat{\theta}_{ij}}{d}\right) = d pen_{ij}'(\hat{\theta}_{ij})$ for any $d \neq 0$.  
The latter requirement can only hold in two scenarios.
First, there is the trivial scenario where $pen_{ij}'(\theta_{ij}) = 0$ for all $\theta_{ij} \neq 0$, that is $pen_{ij}$ is an $L_{0}$ penalty.

Second, if $pen_{ij}'(\hat{\theta}_{ij})= k \neq 0$,
then $pen_{ij}'\left(\frac{\hat{\theta}_{ij}}{d}\right) = d k$.  Treating $\hat{\theta}_{ij}$, and therefore also $k$, as fixed, we denote by $x = \frac{\hat{\theta}_{ij}}{d}$. Then we have $pen_{ij}'\left(x\right) = \frac{\hat{\theta}_{ij} k}{x}$.  It follows that $pen_{ij}(x) = \hat{\theta}_{ij} k \log( | x | ) + c$ for some constant $c$ and $x \neq 0$, that is $pen_{ij}$ is a logarithmic penalty.

This proves that for a regular penalty to be scale invariant it must have $L_0$ or logarithmic $pen_{ij}$.  We now turn our attention to the diagonal penalty.

Let $S$ be some diagonal covariance matrix, and $D$ some diagonal matrix as before.  Let $\hat{\theta}_{ii} = \hat{\Theta}(S)_{ii}$ and $\tilde{\theta}_{ii} = \hat{\Theta}(DSD)_{ii}$.  By scale invariance we must have $\tilde{\theta}_{ij} = \frac{\hat{\theta}_{ij}}{d_{i}^2}$.

Since $S$ is diagonal, it is easy to see that both $\hat{\Theta}(S)$ and $\hat{\Theta}(DSD)$ must also be diagonal, and that $\hat{\theta}_{ii}$ maximises the function:
$$ \log( \theta_{ii} ) - S_{ii} \theta_{ii} - \frac{2}{n} pen_{ii}( \theta_{ii} ), $$
while $\tilde{\theta_{ii}}$ maximises the same function but with $S_{ii}$ replaced by $d_{i}^2 S_{ii}$.  It follows that the corresponding derivatives must both be equal to zero at $\hat{\theta}_{ii}$ and $\tilde{\theta_{ii}}$ respectively ($Pen$ is regular so $pen_{ii}$ is differentiable).  Using this along with $\tilde{\theta}_{ij} = \frac{\hat{\theta}_{ij}}{d_{i}^2}$ we obtain:
$$ pen'_{ii}\left(\frac{\hat{\theta}_{ii}}{d_{i}^{2}}\right) = d_{i}^{2} pen'_{ii}\left( \hat{\theta}_{ii} \right).$$

As before, it follows that $pen_{ii}$ must be either constant or logarithmic.  This proves that for a regular penalty function to be scale invariant it must have either constant or logarithmic penalty on the diagonal entries.

To complete the proof we must show that such penalty functions ($L_0$ or logarithmic off-diagonal penalty and constant or logarithmic diagonal penalty) are always scale invariant.  This follows from Proposition \ref{Prop:PCSepScaleInv} since the $L_0$ and logarithmic penalties are also symmetric PC-separable.

\color{black}

\end{proof}

\begin{proof}[Proof of Proposition \ref{Prop:PCSepScaleInv}]
Let $S$ be a sample covariance matrix and $D$ be a diagonal matrix with non-zero entries $d_{i}$.  Suppose that the estimate $\hat{\Theta}(S)$ decomposes as $\bar{\theta}^{1/2}\bar{\Delta}\bar{\theta}^{1/2}$ and that the estimate $\hat{\Theta}(DSD)$ decomposes as $\tilde{\theta}^{1/2}\tilde{\Delta}\tilde{\theta}^{1/2}$.  To prove scale invariance we need that $\bar{\Delta}=\mathrm{sign}(D)\tilde{\Delta}\mathrm{sign}(D)$ and $\bar{\theta}=D^{2}\tilde{\theta}$.

Since $\hat{\Theta}(S)$ maximises the penalised likelihood at $S$, $\bar{\theta},\bar{\Delta}$ must maximise
\begin{align}\label{eq:PCScaleInv1}
\log(\det(\Theta)) + \sum_{i} \left( \left( 1 - \frac{2c}{n} \right) \log( \theta_{ii} ) - s_{ii} \theta_{ii} \right) - \sum_{i \neq j} \left( s_{ij} \sqrt{\theta_{ii}\theta_{jj}}\Delta_{ij} + \frac{2}{n} pen_{ij}( \Delta_{ij} ) \right),
\end{align}
and similarly, $\tilde{\theta},\tilde{\Delta}$ must maximise
\begin{align}\label{eq:PCScaleInv2}
\log(\det(\Theta)) + \sum_{i} \left( \left( 1 - \frac{2c}{n} \right) \log( \theta_{ii} ) - d_{i}^{2} s_{ii} \theta_{ii} \right) - \sum_{i \neq j} \left( d_{i}d_{j}s_{ij} \sqrt{\theta_{ii}\theta_{jj}}\Delta_{ij} + \frac{2}{n} pen_{ij}( \Delta_{ij} ) \right).
\end{align}
By substituting $\theta'_{ii} = d_{i}^{2} \theta_{ii}$ and $\Delta'_{ij} = \mathrm{sign}(d_{i}d_{j})\Delta_{ij}$ into (\ref{eq:PCScaleInv2}), and noting that $pen_{ij}$ is symmetric about $0$, we get
\begin{align}\label{eq:PCScaleInv3}
\log(\det(\Theta)) + \sum_{i} \left( \left( 1 - \frac{2c}{n} \right) \left( \log( \theta'_{ii} ) - \log( d_{i}^{2} ) \right) - s_{ii} \theta'_{ii} \right) - \sum_{i \neq j} \left( s_{ij} \sqrt{\theta'_{ii}\theta'_{jj}}\Delta'_{ij} + \frac{2}{n} pen_{ij}( \Delta'_{ij} ) \right).
\end{align}
Since $\log(d_{i}^{2})$ is a constant, (\ref{eq:PCScaleInv3}) is of the same form as (\ref{eq:PCScaleInv1}) and they are maximised at the same point.  Hence we have that $\bar{\Delta}=\mathrm{sign}(D)\tilde{\Delta}\mathrm{sign}(D)$ and $\bar{\theta}=D^{2}\tilde{\theta}$.

\end{proof}

\begin{proof}[Proof of Proposition \ref{prop:PriorScaleInv}]

Let $\pi$ be a prior density as given in Proposition \ref{prop:PriorScaleInv}, $S$ be some sample covariance and $D$ some diagonal matrix with non-zero entries.  Writing $L(\Theta \mid S)$ as the likelihood function, $\Theta=\theta^{1/2}\Delta\theta^{1/2}$ and treating $D$ as a constant, the posteriors given $S$ and $DSD$ are
\begin{align}\label{eq:prior1}
\pi( D \Theta D \mid S ) &\propto L( D \Theta D \mid S ) \pi( D \Theta D )  \nonumber \\
&\propto \det(\Delta)^{n/2} \prod_{i} (d_i^2 \theta_{ii})^{\frac{n}{2}} \exp\left( -\frac{n}{2} \sum_{i,j} S_{ij} \sqrt{ d_i^2 \theta_{ii} d_j^2 \theta_{jj} } \Delta_{ij} \right)\prod_{i} (d_i^2 \theta_{ii})^{-c} \prod_{ij} \pi_{ij}(\Delta_{ij}) \mathbb{I}( \Delta \in \mathcal{S}_1 ) \nonumber \\
&= \det(\Delta)^{n/2} \prod_{i} (d_i^2 \theta_{ii})^{\frac{n}{2}-c} \exp\left( -\frac{n}{2} \sum_{i,j} d_id_jS_{ij} \sqrt{ \theta_{ii}  \theta_{jj} } \Delta_{ij} \right) \prod_{ij} \pi_{ij}(\Delta_{ij}) \mathbb{I}( \Delta \in \mathcal{S}_1 )
\end{align}
\begin{align}\label{eq:prior2}
\pi( \Theta \mid DSD ) &\propto L( \Theta \mid DSD ) \pi( \Theta )  \nonumber \\
&\propto \det(\Delta)^{n/2} \prod_{i} \theta_{ii}^{\frac{n}{2}} \exp\left( -\frac{n}{2} \sum_{i,j} d_id_jS_{ij} \sqrt{ \theta_{ii} \theta_{jj} } \Delta_{ij} \right)\prod_{i} ( \theta_{ii})^{-c} \prod_{ij} \pi_{ij}(\Delta_{ij}) \mathbb{I}( \Delta \in \mathcal{S}_1 ) \nonumber \\
&= \det(\Delta)^{n/2} \prod_{i} \theta_{ii}^{\frac{n}{2}-c} \exp\left( -\frac{n}{2} \sum_{i,j} d_id_jS_{ij} \sqrt{ \theta_{ii} \theta_{jj} } \Delta_{ij} \right) \prod_{ij} \pi_{ij}(\Delta_{ij}) \mathbb{I}( \Delta \in \mathcal{S}_1 )
\end{align}

For any measurable set $A$ and $A_D = \{ \Theta : D^{-1} \Theta D^{-1} \in A \}$ the probabilities in Definition \ref{def:PriorScaleInv} can be written as
\begin{align*}
\mathbb{P}_{\pi}\left( \Theta \in A \mid DSD \right) &= \int_{A} \pi\left( \Theta \mid DSD \right) \,d \Theta \\
&= \frac{ \int_A L( \Theta \mid DSD ) \pi( \Theta ) \,d\Theta }{ \int_{\mathcal{S}} L( \Theta \mid DSD ) \pi( \Theta ) \,d\Theta }
\end{align*}
and, noting that $\Theta \in A \iff D \Theta D \in A_D$,
\begin{align*}
\mathbb{P}_{\pi}\left( \Theta \in A_D \mid S \right) &= \int_{A_D} \pi\left( \Theta \mid S \right) \,d \Theta \\
&= \int_{A} \pi\left( D \Theta D \mid S \right) \,d \Theta \\
&= \frac{ \int_A L( D \Theta D \mid S ) \pi( D \Theta D ) \,d\Theta }{ \int_{\mathcal{S}} L( D \Theta D \mid S ) \pi( D \Theta D ) \,d\Theta }
\end{align*}

The result follows by noting that expression (\ref{eq:prior1}) can be obtained by multiplying (\ref{eq:prior2}) by the constant $\prod_{i} (d_i^2)^{\frac{n}{2}-c}$.

\end{proof}

\subsection{Supplementary results for Section \ref{sec:EstEquiv}}\label{subsec:Appendix-EstEqu}

Suppose the value of an estimator $\hat{\theta}=\mbox{diag}(\hat{\Theta})$ and all the entries in $\hat{\Delta}$ are given, except for a pair of partial correlations $(\Delta_{k_1k_2}, \Delta_{k_1k_3})$, for some indexes $k_1,k_2,k_3 \in \{1,\ldots,p\}$.
Suppose that $S$, and the given elements in $\hat{\Delta}$ and $\hat{\theta}$ satisfy the following conditions:

\begin{enumerate}[leftmargin=*, label=(C\arabic*)]
\item $S_{k_1 k_2} / \hat{\theta}_{k_2 k_2}^{-1/2}= S_{k_1 k_3} / \hat{\theta}_{k_3 k_3}^{-1/2}$.
\item $\hat{\Delta}_{k_2 j}=\hat{\Delta}_{k_3 j}$ for all $j \not\in \{k_1,k_2,k_3\}$.
\end{enumerate}

\begin{myprop}\label{prop:pcinvariance_mle}
Under conditions (C1)-(C2) the likelihood function is symmetric in $(\Delta_{k_1k_2}, \Delta_{k_1k_3})$.
\end{myprop}

\begin{proof}

Without loss of generality suppose that the variable indexes are $k_1=1$, $k_2=2$ and $k_3=3$.  The MLE maximises the function
$$ \log( \det( \Theta ) ) - \mathrm{tr}( S \Theta ) = \log( \det( \theta^{1/2}\Delta\theta^{1/2} ) ) - \mathrm{tr}( S \theta^{1/2}\Delta\theta^{1/2} ). $$
Consider this as a function $h(\Delta_{12},\Delta_{13})$ that only depends on $(\Delta_{12},\Delta_{13})$, given a value of the remaining parameters $\hat{\theta}$ and $\hat{\Delta}_{ij}$ for $(i,j) \not\in \{(1,2),(1,3)\}$ satisfying (C1)-(C2).

We shall show that the two terms $\log\det(\Theta)$ and $\mbox{tr}(S\Theta)$ are symmetric in $(\Delta_{12},\Delta_{13})$, when (C1)-(C3) hold.
Using straightforward algebra gives that
\begin{align}
\mbox{tr}(S\Theta)= \mbox{tr}( S \theta^{1/2} \Delta \theta^{1/2} )
=
2 s_{12} \theta_{11}^{1/2} \theta_{22}^{1/2} \Delta_{12} + 2 s_{13} \theta_{11}^{1/2} \theta_{13}^{1/2} \Delta_{13} + c
    \nonumber
\end{align}
where $c$ does not depend on $(\Delta_{12},\Delta_{13})$. Plugging in $\hat{\theta}$ and $\hat{\Delta}_{ij}$ into this expresion and using (C1) gives that is it equal to
\begin{align}
2 \hat{\theta}_{11}^{1/2} s_{12} \hat{\theta}_{22}^{1/2} (\Delta_{12} + \Delta_{13}) + c,
\end{align}
which is symmetric in $(\Delta_{12},\Delta_{13})$.

Consider now $\det(\Theta)$. Using basic properties of the matrix determinant,
\begin{align}
    \det(\Theta)&= \det(\Delta) \prod_{j=1}^p \theta_{jj}
    =|\Delta_{11} - \Delta_{2:p,1} \Delta_{2:p,2:p}^{-1} \Delta_{1,2:p}| |\Delta_{2:p,2:p}| \prod_{j=1}^p \theta_{jj},
    \nonumber
\end{align}
where $\Delta_{i:j,k:l}$ is the submatrix obtained by taking rows $i,i+1,\ldots,j$ and columns $k,k+1,\ldots,l$ from $\Delta$. 
Since $\hat{\theta}$, $\hat{\Delta}_{2:p,2:p}$, and $\hat{\Delta}_{1j}$ for $j \geq 4$ are given, it suffices to show that 
\begin{align}
(\Delta_{12},\Delta_{13},\hat{\Delta}_{14},\ldots,\hat{\Delta}_{1p}) \hat{\Delta}_{2:p,2:p}^{-1} (\Delta_{12},\Delta_{13},\hat{\Delta}_{14},\ldots,\hat{\Delta}_{1p})^T
\label{eq:detDeltainv_term1}
\end{align}
is symmetric in $(\Delta_{12},\Delta_{13})$.
To ease notation let $A=\hat{\Delta}_{2:p,2:p}^{-1}$.
Note that under Condition (C2),
\begin{align}
    \hat{\Delta}_{2:p,2:p}= \begin{pmatrix}
1 & \hat{\Delta}_{23} & \hat{\Delta}_{24} & \ldots & \hat{\Delta}_{2p} \\
\hat{\Delta}_{23} & 1 & \hat{\Delta}_{24} & \ldots & \hat{\Delta}_{2p} \\
\ldots
\hat{\Delta}_{2p} & \hat{\Delta}_{2p} & \hat{\Delta}_{4p} & \ldots & 1
\end{pmatrix}    
\nonumber
\end{align}
and hence
\begin{align}
    \hat{\Delta}_{2:p,2:p}^{-1}= A= \begin{pmatrix}
a_{11} & a_{12} & a_{13} & \ldots & a_{1p-1} \\
a_{12} & a_{11} & a_{13} & \ldots & a_{1p-1} \\
a_{13} & a_{23} & a_{33} & \ldots & a_{3p-1} \\
\ldots
a_{1 p-1} & a_{2 p-1} & a_{3 p-1} & \ldots & a_{p-1 p-3}
\end{pmatrix}.
\nonumber
\end{align}
That is, the first two rows in $A$ are equal, up to permuting the first two elements in each row.
Therefore, \eqref{eq:detDeltainv_term1} is equal to
\begin{align}
    a_{11} \Delta_{12}^2 + a_{11} \Delta_{13}^2 + \sum_{j=3}^{p-1} a_{jj} \hat{\Delta}_{j+1 j+1}^2
    \nonumber \\
    + 2 a_{12} \Delta_{12} \Delta_{13} + 2 \sum_{j=3}^{p-1} a_{1j} \Delta_{12} \hat{\Delta}_{1 j+1}
    + 2 \sum_{j=3}^{p-1} a_{1j} \Delta_{13} \hat{\Delta}_{1 j+1}
    + 2 \sum_{j=3}^p \sum_{k=j+1}^{p} a_{jk} \hat{\Delta}_{j+1 k+1}
    \nonumber \\
    =a_{11} (\Delta_{12}^2 + \Delta_{13}^2)
    + 2 a_{12} \Delta_{12} \Delta_{13} 
    + 2 (\Delta_{12} + \Delta_{13}) \sum_{j=3}^{p-1} a_{1j} \hat{\Delta}_{1 j+1}
    + c',
\nonumber
\end{align}
where $c'$ does not depend on $(\Delta_{12},\Delta_{13})$, which is a symmetric function in $(\Delta_{12},\Delta_{13})$, as we wished to prove.

\end{proof}

Note that because the log-likelihood is a convex function, and therefore has a unique maximum, symmetry in $(\Delta_{k_1k_2}, \Delta_{k_1k_3})$ implies that the MLE will estimate these two partial correlations to be equal.

\begin{mycor}
Under conditions (C1)-(C2) any penalised likelihood with a symmetric PC-separable penalty is symmetric in $(\Delta_{k_1k_2}, \Delta_{k_1k_3})$.
\end{mycor}

\begin{proof}

The proof follows immediately from the proof of Proposition \ref{prop:pcinvariance_mle}, noting that $Pen(\theta,\Delta) = \sum_{i} pen_{ii}(\theta_{ii}) + \sum_{i \neq j} \text{pen}(\Delta_{ij})$ is symmetric in $(\Delta_{12},\Delta_{13})$.

\end{proof}

\begin{mycor}
Under conditions (C1)-(C2) a penalised likelihood with a regular penalty, other than the $L_0$ or logarithmic, is symmetric in $(\Delta_{k_1k_2}, \Delta_{k_1k_3})$ if and only if $\hat{\theta}_{k_2k_2}=\hat{\theta}_{k_3k_3}$.
\end{mycor}

\begin{proof}

From Proposition \ref{prop:pcinvariance_mle} the penalised likelihood is symmetric if and only if $pen_{k_1k_2}\left( \sqrt{\hat{\theta}_{k_1k_1} \hat{\theta}_{k_2k_2}}\Delta_{k_1k_2}\right) + pen_{k_1k_3}\left( \sqrt{\hat{\theta}_{k_1k_1} \hat{\theta}_{k_3k_3}}\Delta_{k_1k_3}\right)$ is symmetric.  Since $Pen$ is regular, this only happens when $\hat{\theta}_{k_2k_2}=\hat{\theta}_{k_3k_3}$ or when $pen_{ij}$ is either $L_0$ or logarithmic.

\end{proof}

\subsection{Proofs for Section \ref{sec:Computation}}\label{subsec:Appendix-CompProofs}

\begin{proof}[Proof of Proposition \ref{prop:Conv1}]

For a fixed $\theta$, optimisation of the penalised likelihood function (\ref{eq:PenLikeFn}) is equivalent to optimisation of the following function
$$ \log(\det(\Delta)) - \sum_{i \neq j} S_{ij} \sqrt{\theta_{ii}\theta_{jj}} \Delta_{ij} - \rho \sum_{i \neq j} \vert \Delta_{ij} \vert. $$
The log-determinant function is known to be concave over the space of positive definite matrices.  For fixed $\theta$ the second term is simply a sum of linear functions.  The third term is simply a sum of clearly concave functions.  Hence the objective function is a sum of concave functions and is therefore concave.

\end{proof}

\newpage
\section{Simulation results}\label{sec:Appendix-Results}

\vspace{10pt}
\begin{table}[ht]
\begin{tabular}{llllll}
\multicolumn{1}{c}{$n=30$} & \multicolumn{1}{c}{FNorm} & \multicolumn{1}{c}{KL} & \multicolumn{1}{c}{MCC} & \multicolumn{1}{c}{Sensitivity} & \multicolumn{1}{c}{Specificity} \\
PC-GLASSO             &     1.42 (0.35)                      &       1.69 (0.58)                 &   0.978 (0.043)                      &     0.999 (0.008)                            &   0.995 (0.010)                              \\
%GLASSO no diag pen             & 1.19 (0.17)               & 1.60 (0.51)            & 0.884 (0.080)           & 0.999 (0.005)                       & 0.970 (0.024)               \\    
%GLASSO no diag pen             & 2.59 (0.36)               & 3.36 (1.42)            & 0.270 (0.044)           & 0.866 (0.105)                   & 0.578 (0.057)                   \\
%GLASSO 3             & 1.49 (0.28)               & 2.01 (0.69)            & 0.789 (0.127)           & 0.996 (0.016)                   & 0.934 (0.054)                   \\
%GLASSO   & 1.75 (0.23)               & 2.34 (0.69)            & 0.677 (0.109)           & 0.998 (0.009)                   & 0.884 (0.062)                   \\
GLASSO   & 2.68 (0.73)               & 3.55 (1.19)            & 0.231 (0.063)           & 0.903 (0.066)                   & 0.477 (0.075)                   \\
%GLASSO diagonals 3   & 1.93 (0.29)               & 2.57 (0.56)            & 0.569 (0.141)           & 0.988 (0.027)                   & 0.810 (0.107)                   \\
%SCAD                & 5.04 (2.97)               & 6.32 (3.97)            & 0.507 (0.179)           & 0.864 (0.110)                   & 0.825 (0.094)                   \\
SCAD               & 8.07 (3.78)               & 10.87 (4.76)           & 0.344 (0.136)           & 0.738 (0.143)                   & 0.764 (0.079)                   \\
%SCAD 3               & 2.12 (0.99)               & 2.78 (1.76)            & 0.527 (0.155)           & 0.982 (0.037)                   & 0.774 (0.128)                   \\
%MCP                 & 6.16 (3.14)               & 7.96 (4.21)            & 0.414 (0.175)           & 0.793 (0.124)                   & 0.788 (0.092)                   \\
MCP                & 8.58 (4.11)               & 11.60 (5.17)           & 0.335 (0.126)           & 0.737 (0.138)                   & 0.756 (0.079)                   \\
%MCP 3                & 2.17 (0.89)               & 2.71 (1.20)            & 0.526 (0.148)           & 0.972 (0.047)                   & 0.783 (0.115)                  
\\
\multicolumn{1}{c}{$n=100$} & \multicolumn{1}{c}{FNorm} & \multicolumn{1}{c}{KL} & \multicolumn{1}{c}{MCC} & \multicolumn{1}{c}{Sensitivity} & \multicolumn{1}{c}{Specificity} \\
PC-GLASSO             &    0.70 (0.11)           &      0.46 (0.12)        &    0.993 (0.017)        &    1 (0)                        &  0.999 (0.004)                               \\
%GLASSO no diag pen             & 0.80 (0.10)               & 0.59 (0.14)            & 0.916 (0.052)           & 1 (0)                           & 0.980 (0.014)                   \\
%GLASSO no diag pen             & 1.66 (0.09)               & 1.20 (0.13)            & 0.304 (0.019)           & 0.996 (0.014)                   & 0.508 (0.031)                   \\
%GLASSO 3             & 0.84 (0.12)               & 0.62 (0.16)            & 0.907 (0.060)           & 1 (0)                           & 0.978 (0.016)                   \\
%GLASSO   & 1.27 (0.15)               & 1.03 (0.25)            & 0.688 (0.090)           & 1 (0)                           & 0.893 (0.047)                   \\
GLASSO   & 1.73 (0.08)               & 1.33 (0.13)            & 0.264 (0.021)           & 0.996 (0.014)                   & 0.433 (0.041)                   \\
%GLASSO diagonals 3   & 1.30 (0.14)               & 1.07 (0.24)            & 0.679 (0.103)           & 1 (0)                           & 0.887 (0.056)                   \\
%SCAD                & 0.83 (0.12)               & 0.58 (0.14)            & 0.850 (0.069)           & 1 (0)                           & 0.961 (0.023)                   \\
SCAD                & 1.33 (0.38)               & 1.01 (0.38)            & 0.739 (0.135)           & 0.958 (0.046)                   & 0.926 (0.049)                   \\
%SCAD 3               & 0.89 (0.13)               & 0.60 (0.15)            & 0.734 (0.087)           & 1 (0)                           & 0.916 (0.038)                   \\
%MCP                & 0.88 (0.19)               & 0.62 (0.18)            & 0.855 (0.085)           & 0.999 (0.005)                   & 0.961 (0.029)                   \\
MCP                 & 1.39 (0.40)               & 1.09 (0.41)            & 0.737 (0.128)           & 0.952 (0.050)                   & 0.928 (0.043)                   \\
%MCP 3                & 0.84 (0.13)               & 0.59 (0.15)            & 0.837 (0.075)           & 1 (0)                           & 0.956 (0.025)                  
\end{tabular}
\caption{Star results}
    \label{tab:StarResults}
\vspace{20pt}

\begin{tabular}{llllll}
\multicolumn{1}{c}{$n=30$} & \multicolumn{1}{c}{FNorm} & \multicolumn{1}{c}{KL} & \multicolumn{1}{c}{MCC} & \multicolumn{1}{c}{Sensitivity} & \multicolumn{1}{c}{Specificity} \\
PCGLasso                   & 1.85 (0.29)               & 2.83 (0.74)            & 0.696 (0.081)           & 0.988 (0.043)                   & 0.917 (0.034)                   \\
%GLasso 1                   & 1.92 (0.24)               & 2.70 (0.58)            & 0.596 (0.083)           & 0.998 (0.011)                   & 0.861 (0.052)                   \\
%GLasso 2                   & 2.26 (0.21)               & 3.11 (0.64)            & 0.469 (0.071)           & 0.998 (0.012)                   & 0.763 (0.066)                   \\
%GLasso 3                   & 2.00 (0.23)               & 2.95 (0.66)            & 0.560 (0.077)           & 0.996 (0.018)                   & 0.839 (0.054)                   \\
%GLasso diagonals 1         & 2.25 (0.37)               & 3.15 (0.75)            & 0.452 (0.083)           & 0.999 (0.009)                   & 0.741 (0.083)                   \\
GLasso         & 2.51 (0.28)               & 3.71 (0.70)            & 0.371 (0.066)           & 0.999 (0.009)                   & 0.644 (0.095)                   \\
%GLasso diagonals 3         & 2.33 (0.26)               & 3.52 (0.65)            & 0.438 (0.081)           & 0.998 (0.012)                   & 0.726 (0.089)                   \\
%SCAD 1                     & 3.96 (1.81)               & 5.57 (2.50)            & 0.493 (0.094)           & 0.964 (0.054)                   & 0.799 (0.068)                   \\
SCAD                      & 7.80 (4.43)               & 11.55 (6.33)           & 0.339 (0.110)           & 0.830 (0.108)                   & 0.715 (0.115)                   \\
%SCAD 3                     & 2.60 (1.37)               & 4.04 (2.18)            & 0.401 (0.098)           & 0.997 (0.014)                   & 0.675 (0.129)                   \\
%MCP 1                      & 4.75 (2.34)               & 6.79 (3.43)            & 0.463 (0.116)           & 0.931 (0.068)                   & 0.786 (0.083)                   \\
MCP                       & 8.22 (4.68)               & 12.30 (6.64)           & 0.329 (0.111)           & 0.821 (0.112)                   & 0.707 (0.125)                   \\
%MCP 3                      & 2.54 (1.68)               & 3.87 (2.69)            & 0.420 (0.092)           & 0.994 (0.022)                   & 0.704 (0.115)                  
      
\\
\multicolumn{1}{c}{$n=100$} & \multicolumn{1}{c}{FNorm} & \multicolumn{1}{c}{KL} & \multicolumn{1}{c}{MCC} & \multicolumn{1}{c}{Sensitivity} & \multicolumn{1}{c}{Specificity} \\
PCGLasso                    & 0.91 (0.15)               & 0.70 (0.20)            & 0.858 (0.069)           & 1 (0)                           & 0.969 (0.019)                   \\
%GLasso 1                    & 1.38 (0.19)               & 0.97 (0.24)            & 0.635 (0.076)           & 1 (0)                           & 0.884 (0.039)                   \\
%GLasso 2                    & 1.63 (0.19)               & 1.11 (0.22)            & 0.483 (0.054)           & 1 (0)                           & 0.778 (0.048)                   \\
%GLasso 3                    & 1.40 (0.20)               & 1.01 (0.24)            & 0.631 (0.074)           & 1 (0)                           & 0.882 (0.038)                   \\
%GLasso diagonals 1          & 1.63 (0.24)               & 1.22 (0.29)            & 0.449 (0.059)           & 1 (0)                           & 0.744 (0.059)                   \\
GLasso          & 1.84 (0.19)               & 1.37 (0.20)            & 0.371 (0.038)           & 1 (0)                           & 0.650 (0.054)                   \\
%GLasso diagonals 3          & 1.65 (0.23)               & 1.26 (0.27)            & 0.449 (0.058)           & 1 (0)                           & 0.743 (0.057)                   \\
%SCAD 1                      & 1.05 (0.16)               & 0.70 (0.16)            & 0.632 (0.053)           & 1 (0)                           & 0.885 (0.029)                   \\
SCAD                      & 0.91 (0.21)               & 0.55 (0.20)            & 0.918 (0.062)           & 0.998 (0.012)                   & 0.984 (0.014)                   \\
%SCAD 3                      & 1.05 (0.16)               & 0.74 (0.15)            & 0.523 (0.045)           & 1 (0)                           & 0.814 (0.037)                   \\
%MCP 1                       & 1.03 (0.16)               & 0.67 (0.16)            & 0.737 (0.068)           & 1 (0)                           & 0.931 (0.026)                   \\
MCP                       & 0.91 (0.22)               & 0.55 (0.22)            & 0.920 (0.066)           & 0.997 (0.014)                   & 0.984 (0.015)                   \\
%MCP 3                       & 1.05 (0.15)               & 0.70 (0.16)            & 0.691 (0.065)           & 1 (0)                           & 0.912 (0.030)                  
               
\end{tabular}
\caption{Hub results}
    \label{tab:HubResults}
\vspace{20pt}
\end{table}

\begin{table}[ht]
\begin{tabular}{llllll}
\multicolumn{1}{c}{$n=30$} & \multicolumn{1}{c}{FNorm} & \multicolumn{1}{c}{KL} & \multicolumn{1}{c}{MCC} & \multicolumn{1}{c}{Sensitivity} & \multicolumn{1}{c}{Specificity} \\
PC-GLASSO             &      3.64 (0.31)                     &      5.26 (0.62)                  &         0.283 (0.093)                &    0.301 (0.194)                             &   0.922 (0.077)                              \\
%GLASSO no diag pen             & 3.79 (0.17)               & 5.64 (0.58)            & 0.235 (0.115)           & 0.157 (0.146)                   & 0.975 (0.038)                   \\
%GLASSO no diag pen             & 3.83 (0.15)               & 5.77 (0.55)            & 0.219 (0.128)           & 0.126 (0.140)                   & 0.984 (0.032)                   \\
%GLASSO 3             & 3.77 (0.16)               & 5.66 (0.56)            & 0.212 (0.104)           & 0.148 (0.137)                   & 0.970 (0.042)                   \\
%GLASSO    & 4.20 (0.20)               & 6.08 (0.75)            & 0.292 (0.090)           & 0.237 (0.138)                   & 0.957 (0.050)                   \\
GLASSO    & 4.27 (0.17)               & 6.63 (0.71)            & 0.258 (0.113)           & 0.162 (0.135)                   & 0.978 (0.041)                   \\
%GLASSO diagonals 3   & 4.08 (0.28)               & 6.05 (0.75)            & 0.268 (0.083)           & 0.297 (0.149)                   & 0.913 (0.084)                   \\
%SCAD                & 4.70 (2.26)               & 7.31 (3.31)            & 0.336 (0.099)           & 0.487 (0.180)                   & 0.844 (0.110)                   \\
SCAD                & 5.98 (4.47)               & 9.17 (5.56)            & 0.290 (0.105)           & 0.444 (0.162)                   & 0.837 (0.114)                   \\
%SCAD 3               & 4.24 (1.38)               & 6.46 (2.45)            & 0.266 (0.085)           & 0.347 (0.219)                   & 0.869 (0.169)                   \\
%MCP                 & 4.78 (1.91)               & 7.57 (2.93)            & 0.326 (0.103)           & 0.464 (0.180)                   & 0.854 (0.097)                   \\
MCP                 & 6.09 (4.61)               & 9.48 (5.87)            & 0.270 (0.105)           & 0.432 (0.159)                   & 0.832 (0.110)                   \\
%MCP 3                & 4.15 (2.29)               & 6.11 (3.33)            & 0.285 (0.085)           & 0.451 (0.224)                   & 0.818 (0.175)                  
\\
\multicolumn{1}{c}{$n=100$} & \multicolumn{1}{c}{FNorm} & \multicolumn{1}{c}{KL} & \multicolumn{1}{c}{MCC} & \multicolumn{1}{c}{Sensitivity} & \multicolumn{1}{c}{Specificity} \\
PC-GLASSO             &     2.30 (0.33)                      &      2.00 (0.38)                  &     0.530 (0.052)                    &  0.855 (0.094)                               &  0.774 (0.069)                               \\
%GLASSO no diag pen & 2.70 (0.35) & 2.20 (0.49)            & 0.525 (0.054)           & 0.825 (0.110)                   & 0.785 (0.089)                   \\
%GLASSO no diag pen             & 2.72 (0.36)               & 2.22 (0.49)            & 0.520 (0.057)           & 0.818 (0.122)                   & 0.785 (0.092)                   \\
%GLASSO 3             & 2.68 (0.34)               & 2.18 (0.45)            & 0.526 (0.053)           & 0.835 (0.098)                   & 0.781 (0.079)                   \\
%GLASSO   & 2.71 (0.45)               & 2.09 (0.52)            & 0.466 (0.064)           & 0.897 (0.089)                   & 0.672 (0.111)                   \\
GLASSO    & 2.70 (0.45)               & 2.10 (0.52)            & 0.462 (0.062)           & 0.903 (0.090)                   & 0.663 (0.112)                   \\
%GLASSO diagonals 3   & 2.72 (0.44)               & 2.15 (0.52)            & 0.468 (0.061)           & 0.893 (0.090)                   & 0.678 (0.112)                   \\
%SCAD                & 1.75 (0.30)               & 1.36 (0.29)            & 0.604 (0.058)           & 0.940 (0.052)                   & 0.785 (0.056)                   \\
SCAD                & 1.60 (0.23)               & 1.33 (0.29)            & 0.767 (0.065)           & 0.908 (0.059)                   & 0.918 (0.039)                   \\
%SCAD 3               & 1.86 (0.37)               & 1.45 (0.29)            & 0.535 (0.049)           & 0.939 (0.054)                   & 0.720 (0.065)                   \\
%MCP                 & 1.70 (0.26)               & 1.33 (0.28)            & 0.701 (0.066)           & 0.929 (0.054)                   & 0.867 (0.046)                   \\
MCP                 & 1.60 (0.23)               & 1.37 (0.31)            & 0.785 (0.065)           & 0.895 (0.062)                   & 0.932 (0.035)                   \\
%MCP 3                & 1.77 (0.34)               & 1.37 (0.31)            & 0.635 (0.059)           & 0.929 (0.060)                   & 0.817 (0.053)                  
\end{tabular}
\caption{AR2 results}
    \label{tab:AR2Results}
\vspace{20pt}

\begin{tabular}{llllll}
\multicolumn{1}{c}{$n=30$} & \multicolumn{1}{c}{FNorm} & \multicolumn{1}{c}{KL} & \multicolumn{1}{c}{MCC} & \multicolumn{1}{c}{Sensitivity} & \multicolumn{1}{c}{Specificity} \\
PC-GLASSO             &    2.30 (0.25)                       &  3.07 (0.51)                      &   0.336 (0.091)                      & 0.310 (0.153)                                &  0.951 (0.041)                               \\
%GLASSO no diag pen             & 2.26 (0.20)               & 3.21 (0.64)            & 0.337 (0.096)           & 0.243 (0.146)                   & 0.973 (0.038)                   \\
%GLASSO no diag pen             & 2.38 (0.18)               & 3.49 (0.61)            & 0.311 (0.118)           & 0.205 (0.158)                   & 0.978 (0.040)                   \\
%GLASSO 3             & 2.31 (0.20)               & 3.33 (0.62)            & 0.278 (0.105)           & 0.224 (0.148)                   & 0.964 (0.038)                   \\
%GLASSO   & 2.65 (0.24)               & 3.57 (0.65)            & 0.356 (0.079)           & 0.349 (0.131)                   & 0.941 (0.052)                   \\
GLASSO    & 2.84 (0.19)               & 4.32 (0.63)            & 0.355 (0.085)           & 0.264 (0.136)                   & 0.969 (0.048)                   \\
%GLASSO diagonals 3   & 2.66 (0.22)               & 3.82 (0.61)            & 0.312 (0.079)           & 0.367 (0.134)                   & 0.915 (0.063)                   \\
%SCAD                & 4.41 (3.95)               & 5.95 (4.48)            & 0.254 (0.094)           & 0.364 (0.131)                   & 0.882 (0.079)                   \\
SCAD               & 4.87 (4.31)               & 6.56 (4.81)            & 0.206 (0.094)           & 0.318 (0.113)                   & 0.876 (0.078)                   \\
%SCAD 3               & 2.70 (0.62)               & 3.89 (0.99)            & 0.310 (0.080)           & 0.373 (0.141)                   & 0.908 (0.080)                   \\
%MCP                 & 4.73 (3.89)               & 6.47 (4.49)            & 0.229 (0.096)           & 0.339 (0.126)                   & 0.879 (0.079)                   \\
MCP                 & 5.12 (3.83)               & 6.98 (4.47)            & 0.194 (0.092)           & 0.320 (0.112)                   & 0.868 (0.078)                   \\
%MCP 3                & 2.46 (0.27)               & 3.31 (0.57)            & 0.316 (0.072)           & 0.427 (0.143)                   & 0.886 (0.075) 
\\
\multicolumn{1}{c}{$n=100$} & \multicolumn{1}{c}{FNorm} & \multicolumn{1}{c}{KL} & \multicolumn{1}{c}{MCC} & \multicolumn{1}{c}{Sensitivity} & \multicolumn{1}{c}{Specificity} \\
PC-GLASSO             &  1.43 (0.16)                         &   1.23 (0.25)                     &   0.572 (0.059)                      &   0.614 (0.110)                              &  0.941 (0.029)                               \\
%GLASSO no diag pen             & 1.57 (0.16)               & 1.33 (0.30)            & 0.569 (0.060)           & 0.606 (0.107)                   & 0.928 (0.017)                   \\
%GLASSO no diag pen             & 1.62 (0.15)               & 1.37 (0.27)            & 0.581 (0.061)           & 0.641 (0.102)                   & 0.941 (0.030)                   \\
%GLASSO 3             & 1.57 (0.16)               & 1.32 (0.28)            & 0.559 (0.062)           & 0.613 (0.115)                   & 0.936 (0.031)                   \\
%GLASSO   & 1.86 (0.21)               & 1.51 (0.33)            & 0.525 (0.064)           & 0.704 (0.101)                   & 0.882 (0.053)                   \\
GLASSO    & 1.93 (0.22)               & 1.64 (0.37)            & 0.526 (0.070)           & 0.724 (0.102)                   & 0.871 (0.065)                   \\
%GLASSO diagonals 3   & 1.86 (0.20)               & 1.54 (0.31)            & 0.514 (0.066)           & 0.706 (0.100)                   & 0.876 (0.049)                   \\
%SCAD                & 1.36 (0.16)               & 1.07 (0.21)            & 0.568 (0.060)           & 0.657 (0.098)                   & 0.925 (0.031)                   \\
SCAD                & 1.32 (0.15)               & 1.08 (0.23)            & 0.598 (0.070)           & 0.610 (0.105)                   & 0.952 (0.029)                   \\
%SCAD 3               & 1.37 (0.15)               & 1.05 (0.19)            & 0.527 (0.066)           & 0.701 (0.085)                   & 0.887 (0.036)                   \\
%MCP                 & 1.34 (0.15)               & 1.07 (0.23)            & 0.597 (0.062)           & 0.625 (0.107)                   & 0.947 (0.028)                   \\
MCP                 & 1.32 (0.14)               & 1.09 (0.22)            & 0.594 (0.070)           & 0.587 (0.110)                   & 0.957 (0.027)                   \\
%MCP 3                & 1.35 (0.16)               & 1.08 (0.23)            & 0.580 (0.067)           & 0.627 (0.100)                   & 0.940 (0.030)    
\end{tabular}
\caption{Random graph results}
    \label{tab:RandomResults}
\end{table}

\newpage
%\printbibliography
\bibliographystyle{plainnat}
\bibliography{PCGLASSO}

\end{document}